%% file: main.tex
\documentclass[sigconf, nonacm]{acmart}

\newcommand\vldbdoi{XX.XX/XXX.XX}
\newcommand\vldbpages{XXX-XXX}
\newcommand\vldbvolume{XX}
\newcommand\vldbissue{X}
\newcommand\vldbyear{2021}
\newcommand\vldbauthors{\authors}
\newcommand\vldbtitle{\shorttitle} 
\newcommand\vldbavailabilityurl{http://vldb.org/pvldb/format_vol14.html}
\newcommand\vldbpagestyle{plain}

\usepackage{subcaption}
\usepackage{graphicx}
\usepackage{balance}
\usepackage[ruled,linesnumbered]{algorithm2e}
\usepackage[utf8]{inputenc}
\usepackage{comment}
\usepackage{subfiles}
\usepackage{xcolor}
\newtheorem{theorem}{Theorem}

\newtheorem{lemma}{Lemma}

\usepackage{hyperref} 
\usepackage{mathtools}

\usepackage[a-2b]{pdfx}
\usepackage{hyperref}
\hypersetup{hidelinks,backref=true,pagebackref=true,hyperindex=true,colorlinks=false,breaklinks=true,urlcolor=black,pdftitle={Title},pdfauthor={Author}}

\begin{document}
\title{SpaceSaving$^\pm$: An Optimal Algorithm for Frequency Estimation and Frequent items in the Bounded Deletion Model}

%%
%% The "author" command and its associated commands are used to define the authors and their affiliations.

\author{Fuheng Zhao}
\affiliation{%
{UC Santa Barbara}
}
\email{fuheng\_zhao@ucsb.edu}

\author{Divyakant Agrawal} 
\affiliation{%
{UC Santa Barbara}
}
\email{agrawal@cs.ucsb.edu}
\author{Amr El Abbadi}
\affiliation{%
{UC Santa Barbara}
}
\email{amr@cs.ucsb.edu}

\author{Ahmed Metwally}
\affiliation{%
{Uber, Inc.}
}
\email{ametwally@uber.com}

%%
%% The abstract is a short summary of the work to be presented in the
%% article.
\begin{abstract}
In this paper, we propose the first deterministic algorithms to solve the frequency estimation and frequent item problems in the \emph{bounded deletion} model. We establish the space lower bound for solving the deterministic frequent items problem in the bounded deletion model, and propose the Lazy SpaceSaving$^\pm$ and SpaceSaving$^\pm$ algorithms with optimal space bound. We develop an efficient implementation of the SpaceSaving$^\pm$ algorithm that minimizes the latency of update operations using novel data structures. The experimental evaluations testify that SpaceSaving$^\pm$ has accurate frequency estimations and achieves very high recall and precision across different data distributions while using minimal space. Our analysis and experiments clearly demonstrate that SpaceSaving$^\pm$ provides more accurate estimations using the same space as the state of the art protocols for applications with up to $\frac{logU-1}{logU}$ of items deleted, where $U$ is the input universe size. Moreover, motivated by prior work, we propose Dyadic SpaceSaving$^\pm$, the first deterministic quantile approximation sketch in the bounded deletion model.
\end{abstract}

\maketitle

%%% do not modify the following VLDB block %%
%%% VLDB block start %%%
\pagestyle{\vldbpagestyle}
\begingroup\small\noindent\raggedright\textbf{PVLDB Reference Format:}\\
\vldbauthors. \vldbtitle. PVLDB, \vldbvolume(\vldbissue): \vldbpages, \vldbyear.\\
\href{https://doi.org/\vldbdoi}{doi:\vldbdoi}
\endgroup
\begingroup
\renewcommand\thefootnote{}\footnote{\noindent
This work is licensed under the Creative Commons BY-NC-ND 4.0 International License. Visit \url{https://creativecommons.org/licenses/by-nc-nd/4.0/} to view a copy of this license. For any use beyond those covered by this license, obtain permission by emailing \href{mailto:info@vldb.org}{info@vldb.org}. Copyright is held by the owner/author(s). Publication rights licensed to the VLDB Endowment. \\
\raggedright Proceedings of the VLDB Endowment, Vol. \vldbvolume, No. \vldbissue\ %
ISSN 2150-8097. \\
\href{https://doi.org/\vldbdoi}{doi:\vldbdoi} \\
}\addtocounter{footnote}{-1}\endgroup
%%% VLDB block end %%%

%%% do not modify the following VLDB block %%
%%% VLDB block start %%%
\ifdefempty{\vldbavailabilityurl}{}{
\vspace{.3cm}
\begingroup\small\noindent\raggedright\textbf{PVLDB Artifact Availability:}\\
The source code, data, and/or other artifacts have been made available at \url{\vldbavailabilityurl}.
\endgroup
}
%%% VLDB block end %%%

\input{motivation}
\input{Background}
\input{SpaceSaving+-}
\input{DSS}

\input{experiments}

\input{conclusion}

\clearpage
%the final balance is required for the pvldb format
\balance

\bibliographystyle{ACM-Reference-Format}
\bibliography{citations}
\clearpage
\input{appendix}

\end{document}

%% file: motivation.tex
\section{Introduction}
With the development of new technologies and advancements in digital devices, massive amounts of data are generated each day and these data contain crucial information that needs to be analyzed. To make the best use of streaming big data, data sketch\footnote{The term sketch refers to the algorithms and data structures that can extract valuable information through {\bf one pass} on the entire data.} algorithms are often leveraged to process the data only once and to provide essential analysis and statistical measures with strong accuracy guarantees while using limited resources. For instance, with limited space and one pass on the dataset, Hyperloglog~\cite{flajolet2007hyperloglog} enables cardinality estimation, the Bloom Filter~\cite{bloom1970space} answers set membership, and KLL~\cite{karnin2016optimal, ivkin2019streaming} provides quantile approximation.

\textcolor{black}{Two fundamental problems in data sketch research are identifying the most frequently occurring items, a.k.a. frequent items, heavy hitters, Top-K, elephants, and iceberg problem, and estimating the frequency of an item, a.k.a the frequency estimation problem. The formal definition of these two problems are included in Section~\ref{subsection definition}.} Several algorithms~\cite{charikar2002finding, manku2002approximate, cormode2005improved, metwally2005efficient} have been proposed to solve both problems with tunable accuracy guarantees using small memory footprints. These algorithms can be categorized into \emph{counter} based and \emph{linear sketch} based approaches. The counter based approach~\cite{metwally2005efficient} tracks a subset of input items and their estimated frequencies. The linear sketch based approach~\cite{charikar2002finding,cormode2005improved, jin2003dynamically} tracks attribute information from the universe. While linear sketches~\cite{charikar2002finding,cormode2005improved} directly solve the frequency estimation problem, they require additional structures such as heaps or need to impose hierarchical structures over the assumed-bounded universe to solve the frequent items problem.
%. when assuming a bounded universe, can be used to find frequent items by querying all potential items in the universe and identifying the most frequent items. 
%Note this approach might take significant time and hence often a tree structure is used to speed up the searching time~\cite{cormode2008finding}. In addition, t
The frequency estimation and frequent items problems have important applications, such as click stream analysis~\cite{gunduz2003web, metwally2005efficient, das2009cots}, distributed caches~\cite{zakhary2020cot}, database management~\cite{cormode2004holistic, pike2005interpreting,fang1999computing, ting2018data}, and network monitoring~\cite{sivaraman2017heavy, basat2020designing, harrison2020carpe}. In addition, if inputs are drawn from a bounded universe, frequency estimation sketches can also solve the quantile approximation problem~\cite{cormode2005improved,wang2013quantiles}.

Previously, all sketches assumed the \textit{insertion-only} model or the \textit{turnstile} model. The insertion-only model consists only of insert operations, whereas the turnstile model consists of both insert and delete operations such that deletes are always performed on previously inserted items~\cite{wang2013quantiles}. Supporting both insert and delete operations is harder, e.g., sketches in the turnstile model incur larger space overhead and higher update times compared to sketches in the insertion-only model. Jayaram et al.~\cite{jayaram2018data} observed that in practice many turnstile models only incur a fraction of deletions and proposed an intermediate model, the \textit{bounded deletions} model, in which at most $(1-\frac{1}{\alpha})$ of prior insertions are deleted where $\alpha \geq 1$ and $(1-\frac{1}{\alpha})$ upper bounds the delete:insert ratio. Setting $\alpha$ to 1, the bounded deletion model becomes the insertion-only model. 
%Jayaram et al.~\cite{jayaram2018data} identified the bounded deletion model as particularly useful for applications such as computing differences between network traffic patterns, $L_0$ estimation of moving sensors (monitoring wildlife or water flow pattern), etc. Many works have proposed algorithms for discovering novel properties of streaming tasks in the bounded deletion model and have analysed their space and time complexity~\cite{braverman2020coin,kallaugher2020separations, zhaokll}. 

The bounded deletion model is important in many real-world applications such as summarizing product sales in electronic commerce platforms and rankings in standardized testing. Many companies use purchase frequency to check if their customers are satisfied with a product and to identify important groups for advertising and marketing campaigns. After customers purchase products, a certain percentage of the purchases may be returned and the frequency estimation should reflect these changes. However, for any financially viable company, it is highly unlikely that all of these customers will return their purchases and hence in most cases the bounded deletion model can be assumed. In the context of standardized testing such as SAT, ACT, and GRE, frequency estimations are often used to compare and contrast performance among different demographics\footnote{https://reports.collegeboard.org/pdf/2020-total-group-sat-suite-assessments-annual-report.pdf}. Students may request regrades of their exams only once to rectify any machine errors or human errors. Hence, the bounded deletion model can be used with $\alpha=2$. \textcolor{black}{Recently, the bounded deletion model has gained in popularity, and several algorithms have been proposed to discover novel properties of streaming tasks~\cite{jayaram2018data,braverman2020coin,kallaugher2020separations, zhaokll} in this model.} 

\textcolor{black}{In this paper, we present the SpaceSaving$^\pm$ algorithm that solves both frequency estimation and frequent items problems in the bounded deletion model with state of the art performance and minimal memory footprint.} If the administrator of a large data set knows, a priori, that deletions are not arbitrarily frequent compared to insertions, then SpaceSaving$^\pm$ can efficiently capture these changes and identify frequent items with small space, fast update time, and high accuracy. In addition, inspired by quantile summaries~\cite{gilbert2002summarize,cormode2005improved,wang2013quantiles}, we further demonstrate how to leverage SpaceSaving$^\pm$ to support deterministic quantile approximation in the bounded deletion model. In summary, the main contributions of this paper are: (i) we present Lazy SpaceSaving$^\pm$ and SpaceSaving$^\pm$, two space optimal deterministic algorithms in the bounded deletion model and establish their space optimality and correctness; (ii) we propose the Dyadic SpaceSaving$^\pm$ sketch, the first deterministic quantile approximation sketch in the bounded deletion model; (iii) we implement SpaceSaving$^\pm$ using two heaps to minimize the update time; and (iv) we evaluate SpaceSaving$^\pm$ and compare it to state of the art approaches~\cite{cormode2005improved, charikar2002finding, jayaram2018data} and achieve 5 orders of magnitude better accuracy on real-world dataset.

%Count-Min~\cite{cormode2005improved}, Count-Median~\cite{charikar2002finding}, CSSS Sketch~\cite{jayaram2018data} using synthetic and real datasets.

The paper is organized as follows, Section~\ref{sec-back} discusses the background of frequency estimation and frequent items problem, and gives an overview of previous algorithms. Section~\ref{sec-alg} introduces Lazy SpaceSaving$^\pm$ and SpaceSaving$^\pm$ in the bounded deletion model, demonstrates that these algorithms are space optimal, and presents an efficient implementation using a min heap and a max heap data structure to minimize update time. Section~\ref{sec-dydadic} introduces the Dyadic SpaceSaving$^\pm$ quantile sketch that extends SpaceSaving$^\pm$ to solve the deterministic quantile approximation problem in the bounded deletion model. Section~\ref{sec-eval} shows the experimental evaluations conducted using synthetic and real world datasets and compares SpaceSaving$^\pm$ to the state of the art sketches that support delete operations. Finally, Section~\ref{sec-concl} summarizes our contributions and concludes this work.

%% file: Background.tex
\section{Background}
\label{sec-back}
\begin{table*}[tbph]
    \centering
    \begin{tabular}{|c||c|c|c|c|c|c}\hline
        {\bf Sketch} & {\bf Space}& {\bf Update Time} & {\bf Randomization} & {\bf Model} & {\bf Note}   \\ \hline
        \hline
        
        SpaceSaving~\cite{metwally2005efficient} & $O(\frac{1}{\epsilon})$ & $O(log\frac{1}{\epsilon})$ & Deterministic& Insertion-Only & see Lemma~\ref{lemma: spacesaving overestimation}\\ \hline
        
        Count-Min~\cite{cormode2005improved} & $O(\frac{1}{\epsilon}log\frac{1}{\delta})$ & $O(log\frac{1}{\delta})$ & Randomized & Turnstile & Never Underestimate \\ \hline
        
        Count-Median~\cite{charikar2002finding} & $O(\frac{1}{\epsilon}log\frac{1}{\delta})$ &  $O(log\frac{1}{\delta})$& Randomized & Turnstile & Unbiased Estimation \\ \hline
        
        CSSampSim~\cite{jayaram2018data} & $O(\frac{1}{\epsilon}log\frac{1}{\delta}log\frac{\alpha logU}{\epsilon})$ bits&  $\Theta (\frac{ \alpha logU}{\epsilon U}log\frac{1}{\delta})$ & Randomized & Bounded Deletion & \\ 
        \hline
        
        Lazy-SpaceSaving$^\pm$ & $O(\frac{\alpha}{\epsilon})$ &  $O(log\frac{\alpha}{\epsilon})$ & Deterministic & Bounded Deletion & see Lemma~\ref{lemma: lazy nvr underestimate}\\ \hline
        
        SpaceSaving$^\pm$ & $O(\frac{\alpha}{\epsilon})$ &  $O(log\frac{\alpha}{\epsilon})$ & Deterministic & Bounded Deletion & \\ \hline

    \end{tabular}
    \caption{Comparison between different $l_1$ frequency estimation algorithm.}
    \label{tab:sketch-comparison}
\end{table*}
\subsection{Preliminaries} \label{subsection definition}

\begin{table}[htbp]
    \begin{center}% used the environment to augment the vertical space
    % between the caption and the table
    \begin{tabular}{r l }
    \toprule
    $\sigma$ & Data stream\\
    $N$ & Data stream length\\
    $universe$ & All data are drawn from $universe$\\
    $U$ & Size of the $universe$\\
    $F$ & Frequency vector\\
    %$|F|_l$ & $l$-norm of frequency vector $F$\\
    %$m$ & Maximum entry in $F$\\
    $f(x)$ & $x$'s true frequency\\ 
    $\hat{f}(x)$ & $x$'s estimated frequency\\ 
    $\pi(cond)$ & Return 1 if $cond$ is true, 0 otherwise \\
    $\epsilon$ & Accuracy\\ 
    $\delta$ & Failure probability\\ 
    $I$ & Number of insertions\\
    $D$ & Number of deletions\\
    $\alpha$ & In the bounded deletion model, $D\leq(1 - \frac{1}{\alpha})I$ \\
    $minCount$ & The minimum count in SpaceSaving\\
    $minItem$ & The item associated with $minCount$\\
    \bottomrule
    \end{tabular}
    \end{center}
    \caption{\textcolor{black}{Table of symbols}}
    \label{tab:Table of Symbols}
\end{table}

Given a stream $\sigma = \{ item_t \}_{t=\{1,2,...,N\}}$ of length $N$ and items drawn from $universe$ of size $U$, the frequency of an item $x$ is $f(x) = \sum_{t=1}^{N} \pi(item_t = x)$ where $\pi$ returns 1 if $item_t$ is $x$, and returns 0 otherwise. The stream $\sigma$ implicitly defines a frequency vector $F = \{f(a_1),...,f(a_u)\}$ for items $a_1$,...,$a_U$ in the $universe$. Some algorithms assume the $universe$ is bounded, such as in the IP network monitoring context~\cite{sivaraman2017heavy}. Many algorithms assume unit updates, such as the click stream, while others consider the scenario of weighted updates such as purchasing multiple units of the same item at once on e-commerce platform. In this paper, we focus on the unit updates model and assume that items cannot be deleted if they were not previously inserted and hence all entries in frequency vector $F$ are non-negative.
%as weighted updates can be expanded into multiple unit updates. 
%In the insertion-only model, all $weight_t$ are assumed to be strictly positive. Whereas in the turnstile model, $weight_t$ can be negative. We assume that items cannot be deleted if they were not previously inserted and hence all entries in frequency vector $F$ are non-negative.
%There are also two categories of turnstile model: strict and general. In the {\em strict turnstile model}, deletes are performed on previously inserted items, hence all entries in $F$ are non-negative; in the {\em general turnstile model}, entries in $F$ can potentially be negative~\cite{wang2013quantiles}. 

The frequency estimation problem takes an accuracy parameter $\epsilon$ and estimates the frequency of any item $x$ such that $|\hat{f}(x) - f(x)| \leq \epsilon \cdot |F|_{p}$, where $p$ can be either 1 or 2 corresponding to $l_1$ or $l_2$ norm and respectively provide $l_1$ or $l_2$ guarantees, $\hat{f}(x)$ is the estimated frequency and $f(x)$ is the actual frequency. When $p > 2$, providing $l_p$ guarantee requires $poly(U)$ space~\cite{bar2004information}. In this paper, we focus on the $l_1$ problem variation. \textcolor{black}{The $\phi$ frequent items problem is to identify a bag of \textit{heavy items} whose frequency is greater or equal to the specified threshold $\phi \cdot |F|_1$, where $0<\phi<1$. These heavy items are also known as the hot items.\footnote{The term “Hot Items” was coined in~\cite{cormode2005s}}} In addition, some algorithms solve the $(\epsilon, \phi)$-approximate frequent items problem, which is to identify a bag of items $B$, given parameter $0<\epsilon \leq \phi<1$, such that $B$ does not contain any element with frequency less then $(\phi - \epsilon) |F|_1$, i.e., $\forall i \in B$, $f(i) > (\phi - \epsilon) |F|_1$ and $B$ contains all items with frequency larger than $\phi |F|_1$ i.e., $\forall i \not\in B$, $f(i) < \phi |F|_1$.
%\footnote{The $l$ norm of frequency vector $F$ is $(\sum_{i=item_1}^{item_U}|f(i)|^{l})^{1/l}$.} 

\subsection{Deterministic and Randomized Solutions}
Reporting the exact frequent items requires $\Omega(N)$ space~\cite{cormode2008finding}. With limited memory and large universes, solving the exact frequent items problem is infeasible. An alternative and more practical approach in the context of big data is to use approximation techniques.

\textcolor{black}{Deterministic solutions for the $\phi$ frequent items problem guarantee to return all heavy items and potentially some light-weighted items~\cite{misra1982finding, demaine2002frequency, karp2003simple, metwally2005efficient}. Randomized solutions for the $(\epsilon, \phi)$-approximate frequent items problem allow the algorithm to fail with some probability $\delta$~\cite{charikar2002finding, cormode2005improved, jayaram2018data}. In much of the literature, the failure probability is set to $\delta = O(U^{-c})$ where $U$ is the bounded universe size and $c$ is some constant. From the user perspective, deterministic algorithms provide stronger guarantees as all heavy items are identified. Randomized algorithms, on the other hand, make a best effort to report all heavy items and do not report any light weighted items.}

\subsection{Algorithms in Insertion Only Model} \label{Section 2.3}

The insertion-only model consists only of insert operations and many of the proposed algorithms in the insertion-only model are counter-based algorithms which maintain a fixed number of $item$ and $count$ pairs, and the underlying maintenance algorithm increments or decrements these counts to capture the frequency of items that are being tracked.

The first counter-based one pass algorithm to find the most frequent items in a large dataset dates back to the deterministic \textbf{Majority} Algorithm by Boyer and Moore in 1981~\cite{boyer1991mjrty}. 
%The algorithm uses a single counter that stores an item and a corresponding count. It can be described as follows: store the first item and initialize its count to 1. For all subsequent items, if the new item is stored then increase the count by 1; if new item is not stored and the counter is 0, then store the new item and update the count to 1; else reduce the count by 1. This algorithm guarantees that if there is a majority item, then this item will be discovered after processing all items. This algorithm solves the frequent item problem where $\epsilon = \phi = \frac{1}{2}$. For the rest of this paper, we assume $\epsilon=\phi$ unless explicitly stated otherwise. 
In 1982, Misra and Gries \cite{misra1982finding} generalized the majority problem and proposed the deterministic \textbf{MG} summary which uses $O(\frac{1}{\epsilon})$ space to solve the frequency estimation and $\phi$ frequent items problems. \textcolor{black}{MG summary is a set of $\frac{1}{\epsilon}$ counters that correspond to monitored items. When a new item arrives, MG performs the following updates: if the new item is monitored, then increase its count by 1. Else if the summary is not full, monitor the new item. Else decrement all counts by 1 and remove any items with a count of zero. As a result of MG decrementing all counts by 1 when an arriving item is unmonitored,  MG always underestimate item's frequency and a straightforward implementation requires $O(1/\epsilon)$ update time.} Two decades later, Manku and Motwani \cite{manku2002approximate} proposed a randomized \textbf{StickySampling} algorithm and a deterministic \textbf{LossyCounting} algorithm with worst case space $O(\frac{1}{\epsilon}log(\epsilon N))$, which exceeds the memory cost of MG summary. In 2003, Demaine et al.~\cite{demaine2002frequency} and Karp et al.~\cite{karp2003simple} independently generalized the majority algorithm and proposed the \textbf{Frequent} algorithm, which is a rediscovery of the MG summary.

Two years later, in 2005, Metwally, Agrawal, and El~Abbadi~\cite{metwally2005efficient} proposed the \textbf{SpaceSaving} algorithm that provides highly accurate frequency estimates for frequent items and also presents a very efficient method to process insertions. SpaceSaving SpaceSaving uses $k$ counters to store an \textcolor{black}{item's identity, estimated count and estimation error information, i.e., $(item, count_{item}, error_{item})$, and $error_{item}$ is an upper bound on the difference between the item's estimated frequency and its true frequency.} When $k=\frac{1}{\epsilon}$, SpaceSaving solves both frequency estimation and frequent items problem. As shown in Algorithm~\ref{Space Saving algorithm}, insertions proceed as follows, when a new item ($newItem$) arrives: if $newItem$ is monitored, then increment its count; if $newItem$ is not monitored and sketch size not full, then monitor $newItem$, and set $count_{newItem}$ to 1 and $error_{newItem}$ to 0; otherwise, SpaceSaving replaces the item ($minItem$) with the minimum count ($minCount$) by $newItem$, sets $error_{newItem}$ to $minCount$ and increments $count_{newItem}$. \textcolor{black}{In the original SpaceSaving~\cite{metwally2005efficient}, $error_{item}$ is only used to show certain properties of the algorithm, while in this work we leverage this information for handling deletions. As shown in Algorithm~\ref{Space Saving query}, to estimate the frequency of an item in SpaceSaving, if the item is inside the sketch then report its count value, otherwise report 0. In~\cite{agarwal2012mergeable}, Agrawal et al. showed that both SpaceSaving and MG are mergeable~\footnote{Mergeability is desired for distributed settings and means summaries over different datasets can be merged into a single summary as if the single summary processed all datasets.}, and a SpaceSaving algorithm with $k$ counters can be isomorphically transformed into MG summary with $k-1$ counters. Although SpaceSaving and MG share similarities, they follow different sets of update rules. When a new inserted item is unmonitored and the sketch is full, SpaceSaving algorithm replaces the min item with the new item and increments the count by one, whereas the MG decrements all item counts' by 1. As a result, SpaceSaving maintains an upper bound on the frequency of stored items, while the MG always underestimates the frequency. Since SpaceSaving always increments one of the counts by one, the sum of all counts in SpaceSaving is equal to the $|F|_1$. Moreover, the SpaceSaving algorithm elegantly handles the case when an unmonitored new item arrives and the sketch is full, and naturally leads to a min-heap implementation such that incrementing any count and replacing the min item have $O(log(k))$ update times, where $k$ is the number of counters. SpaceSaving can also be implemented with a linked list data structure by keeping items with equal counts in a group, resulting in an $O(1)$ update time~\cite{metwally2005efficient}.}

SpaceSaving satisfies the following properties (the first three properties are proved in~\cite{metwally2005efficient} while the latter two are proved in Appendix~A):

\begin{lemma} \label{lemma: spacesaving overestimation}
\textcolor{black}{Frequency estimations for monitored items are never underestimated in SpaceSaving.}
\end{lemma}

\begin{lemma} \label{important lemma1}
\textcolor{black}{SpaceSaving with $k=\frac{1}{\epsilon}$ counters ensures that after processing $I$ insertions, the minimum count of all monitored items is no more than $\frac{I}{k}=\epsilon I$, i.e, $minCount < \epsilon I$.}
\end{lemma}

\begin{lemma} \label{important lemma2}
\textcolor{black}{All items with frequency larger than or equal to $minCount$ are inside the SpaceSaving sketch.}
\end{lemma}

\begin{lemma} \label{lemma: SS estimation error upper}
\textcolor{black}{The sum of all estimation errors is an upper bound on the sum of the frequencies of all unmonitored items.}
\end{lemma}

\begin{lemma} \label{lemma error bound}
\textcolor{black}{SpaceSaving with $\frac{1}{\epsilon}$ counters can estimate the frequency of any item with an additive error less than $\epsilon I$.}
\end{lemma}

\textcolor{black}{Lemma~\ref{important lemma1} and Lemma~~\ref{important lemma2}, show that SpaveSaving with $\frac{1}{\epsilon}$ counters reports all items whose frequencies are larger than or equal to $\epsilon |F|_{1}$. Empirically, many papers have demonstrated that SpaceSaving outperforms other deterministic algorithms and it is considered to be the state of the art for finding frequent items~\cite{cormode2008finding, manerikar2009frequent}. Moreover, due to the superior performance of SpaceSaving, many works use it as a fundamental building block~\cite{sivaraman2017heavy,ting2018data,zakhary2020cot,basat2020designing,zhang2021cocosketch}}. Recently, a new randomized algorithm \textbf{BPtree} was proposed by Braverman et al.~\cite{braverman2017bptree} to solve the frequent items problem with $l_2$ guarantees in the insertion-only model using $O(\frac{1}{\epsilon^2}log\frac{1}{\epsilon})$ space.

\begin{algorithm}[]
\SetAlgoLined
    %$Sketch \leftarrow  (\emptyset$, $\frac{1}{\epsilon}$)\; 
    
    %// $\frac{1}{\epsilon}$ is the sketch size
    \For{ item from insertions}{
        \uIf{item $\in$ Sketch}{
            $count_{item}$ += 1 \;
        }
        \uElseIf{Sketch not full}{
            Sketch = Sketch $\cup$ item \;
            
            $count_{item}$ = 1 \;
            $error_{item}$ = 0 \;
        }
        \uElse{
            // Sketch is full\;
            $minItem$ = $min_{minItem \in Sketch}$ $count_{minItem}$\;
            $error_{item}$ = $count_{minItem}$ \;
            $count_{item}$ = $count_{minItem}$+1 \;
            Replace $(minItem, count_{minItem}, error_{minItem})$ by $(item, count_{item}, error_{item})$ 
            
        }
    }
 \caption{SpaceSaving Insert Algorithm}
 \label{Space Saving algorithm}
\end{algorithm}

\begin{algorithm}[]
\SetAlgoLined

        \uIf{$item$ $\in$ Sketch}{
            return $count_{item}$
        }
        return 0\;

 \caption{SpaceSaving Query(item)}
 \label{Space Saving query}
\end{algorithm}

\subsection{Algorithms in Turnstile Model}
In turnstile model, the stream consists of both insert and delete operations such that the deletes are always performed on previously inserted items. The data sketches for solving the frequency estimation problem in the turnstile model are known as \textbf{linear sketches}~\cite{cormode2008finding}. 
%A linear sketch is a linear projection of the frequency vector $F$, in which it can be seen as the product of the vector $F$ with a matrix. 
While the counter-based solutions solve both the frequency estimation and frequent items problems, the linear sketch solutions directly answer the frequency estimation problem but need additional information to solve the frequent items problem. In general, linear sketch algorithms assume the input comes from a bounded universe and assume the maximum entry in the frequency vector $F$ is $O(poly(U))$. When assuming a bounded universe, linear sketches can query all potential items in the universe to identify the frequent ones. 

In 1999, Alon et al.~\cite{alon1999space} proposed the randomized \textbf{AMS} sketch to approximate the second frequency moment. Charikar et~al.~\cite{charikar2002finding} improved upon the AMS sketch and proposed the \textbf{Count-Median} sketch, a randomized algorithm that summarizes the dataset. The Count-Median sketch provides an unbiased estimator and uses $O(\frac{1}{\epsilon}log\frac{1}{\delta})$ space to solve the $l_{1}$ frequency estimation problem and uses $O(\frac{1}{\epsilon^2}log\frac{1}{\delta})$ space to solve the $l_{2}$ frequency estimation problem. Later, Cormode and Muthukrishnan~\cite{cormode2005improved} proposed the \textbf{Count-Min} sketch that shares a similar algorithm and data structure as the Count-Median sketch. Count-Min sketch never underestimates frequencies, and uses $O(\frac{1}{\epsilon}log\frac{1}{\delta})$ space to solve the $l_{1}$ frequency estimation problem.

Although one may exhaustively iterate through the universe to find frequent items,
iterating through the universe can be slow and inefficient. As a result, Cormode and Muthukrishnan~\cite{cormode2005improved} suggested to imposes a hierarchical structure on the bounded universe, such that there are $log_{2}U$ layers and one Count-Min or Count-Median sketch per level and then use divide-and-conquer to search for the frequent items from the largest range to an individual item. 
%If a dyadic range has high frequency estimation then it will search for the sub-dyadic interval ranges until a frequent single item is found. 
The required space is $O(\frac{1}{\epsilon}log\frac{1}{\delta}logU)$ and update time is $O(log\frac{1}{\delta}logU)$. Dyadic interval is in the form of $[i2^j, (i+1)2^j-1]$ for $j \in log_{2}U$ and any constant $i$, such that any ranges can be decomposed into at most $log_{2}U$ disjoint dyadic ranges~\cite{cormode2019answering}. Dyadic intervals over a bounded universe can be integrated with frequency estimation sketches to solve the quantile approximation problem in the turnstile model~\cite{gilbert2002summarize, cormode2005improved, wang2013quantiles}

%\begin{figure}[tbph]
%\centering
%\includegraphics[scale=0.18]{figures/SpaceSaving Count Median Sketch.png}
%\caption{Underlying Data Structure of Count-Median and Count-Min Sketch}
%\label{fig:count sketch, Count-Min sketch}
%\end{figure}

\subsection{Algorithms in Bounded Deletion Model}
In the bounded deletion model, the stream consists of both insert and delete operations and a constant $\alpha\geq1$ is given such that at most $(1-\frac{1}{\alpha})$ of prior insertions are deleted, i.e., $D \leq (1-\frac{1}{\alpha})I$, where $I$ is the number of insertions and $D$ is the number of deletions.
Jayaram et al.~\cite{jayaram2018data} proposed the CSSS (\textbf{C}ount-Median \textbf{S}ketch \textbf{S}ample \textbf{S}imulator) algorithm to solve the frequency estimation problem in the bounded deletion model. The Count-Median and Count-Min sketches require $O(\frac{1}{\epsilon}log\frac{1}{\delta})$ number of counters. Assuming $\delta = O(U^{-c})$ for some constant $c$ and the maximum entry of $F$ is $O(poly(U))$, then these two sketches require $O(\frac{1}{\epsilon}log^2(U))$ bits, which achieves the optimal lower bound in the turnstile model~\cite{jowhari2011tight}. Jayaram et al.~\cite{jayaram2018data} pointed out that in the bounded deletion model by simulating the Count-Median sketch on $poly(\alpha logU/\epsilon)$ uniformly sampled items from a stream and scaling the weights at the end, the $CSSS$ sketch can accurately approximate the true frequency of an item with high probability. Hence, by carefully tuning the size of the Count-Median sketch, CSSS requires $O(\frac{1}{\epsilon}log\frac{1}{\delta}log\frac{\alpha logU}{\epsilon})$ bits, improving the overall space compared to sketches in the turnstile model.

\subsection{Summary}

In Table \ref{tab:sketch-comparison}, we compare the differences and similarities among several different sketches for $l_1$ frequency estimation. These sketches can also solve $l_1$ heavy hitters, though some sketches may need additional modifications to the parameters or leverage external data structures. In Table \ref{tab:Table of Symbols}, we listed the important symbols used in the paper. 
Counter-based solutions have many advantages over linear sketches. Counter-based solutions are guaranteed to report all heavy items; they use $O(log\frac{1}{\epsilon})$ update time instead of $O(logU)$ update time where $\frac{1}{\epsilon}$ is often less than the universe size $U$; and they make no assumptions on the $universe$ and thus can be useful in Big Data applications where items are drawn from unbounded domains. In this paper, we present SpaceSaving$^\pm$, an optimal counter-based deterministic algorithm with $l_1$ guarantee to solve both the frequency estimation and frequent items problem in the bounded deletion model using $O(\frac{\alpha}{\epsilon})$ space. 

%The SpaceSaving$^\pm$ sketch builds upon SpaveSaving~\cite{metwally2005efficient}, the state of the art solution in the insertion-only model. 

%% file: SpaceSaving+-.tex
\section{The SpaceSaving$^\pm$ Algorithm}
\label{sec-alg}
In this section, we first show the space lower bound for solving the $\phi$ frequent items problem in the bounded deletion model. Then, we introduce the Lazy $SpaceSaving^\pm$ and $SpaceSaving^\pm$ algorithms with optimal space to solve both the frequency estimation and frequent items problems in the bounded deletion model \textcolor{black}{in which the total number of deletions ($D$) is less than $(1 - \frac{1}{\alpha})$ of the total insertions ($I$) where $\alpha \geq 1$}. \textcolor{black}{Given a user specified accuracy on the parameter $\epsilon$, a deterministic algorithm for frequency estimation and frequent items problems must:}
\begin{itemize}
    \item Approximate the frequency of all items $i$ with high accuracy such that $\forall~i: |f(i) - \hat{f}(i)|) \leq \epsilon |F|_{1}$; and 
    \item Report all the items with frequency greater than or equal to $\epsilon |F|_{1}$.
\end{itemize}
We propose Lazy SpaceSaving$^\pm$ and SpaceSaving $^\pm$. % to solve the frequency estimation and frequent items problems in the bounded deletion model. 
The main difference between the two variants of SpaceSaving$^\pm$ arises in the way deletions are handled. Since we assume the strict bounded deletion model, a delete operation must correspond to a previously inserted item. If the item is being tracked in the sketch, processing such a delete operation is straightforward since the count associated with the item can be decreased by 1. On the other hand, the challenge arises when the sketch maintenance algorithm encounters a delete of an item that is not being tracked in the sketch. We develop different ways of handling such a delete in the two algorithms and the resulting correctness guarantees.

\subsection{Space Lower Bound}

We first show that there is no \textcolor{black}{counter based algorithm that can solve the deterministic frequent items problem in the bounded deletion model using less than $\frac{\alpha}{\epsilon}$ counters.} 

\begin{theorem} 
In the bounded deletion model, any counter based algorithm needs at least $\frac{\alpha}{\epsilon}$ counters to solve the deterministic frequent items problem.
\end{theorem}

\begin{proof} By Contradiction.

\textcolor{black}{Assume that there exists a counter based deterministic solution using $k<\frac{\alpha}{\epsilon}$ counters that can report all the items with frequency larger than or equal to $\epsilon|F|_1$. Consider a stream $\sigma$ with bounded deletions that contains $I$ insertions and $D$ deletions where all insertions come before any deletions. Let the $I$ insertions consist exactly of $\frac{\alpha}{\epsilon}$ unique items, each with an exact count of $\frac{\epsilon}{\alpha}I$. After processing all insertions, the optimal algorithm with $k<\frac{\alpha}{\epsilon}$ counters will monitor at most $k$ unique items, and there would be at least one item from the insertions that is left out. Let the set $Missing$ contains all such unique items that appeared in $I$ but are not monitored by the optimal algorithm. Now let the $D = (1-\frac{1}{\alpha})I$ deletions be applied arbitrarily on the monitored items. After all $D$ deletions, all items in $Missing$ have frequency of $\frac{\epsilon}{\alpha}I$ in which $\frac{\epsilon}{\alpha}I = \epsilon (I-D) \geq \epsilon|F|_1$, and these items are frequent and must be monitored by the optimal algorithm. However, the sketch, with space $k$, after processing all insertions loses the information regarding $Missing$. Therefore, it is not possible to use less than $\frac{\alpha}{\epsilon}$ counters to solve the deterministic frequent items problem in the bounded deletion model.}
\end{proof}

\subsection{Lazy SpaceSaving$^\pm$ Approach}
Since supporting both insertions and bounded deletions is a much harder task compared to only allowing for insertions, the overall space bound needs to be increased. From the previous section, we can see that if the goal is to report all the items with frequency more than $\epsilon|F|_{1}$ times, where $|F|_{1} = I - D$, we need to track more items. Since before any deletions, the sketch has no knowledge regarding which items are going to be deleted, then all elements with frequency higher than $\epsilon (I-D)$ are potential candidates before any deletions. We at least need an algorithm that can identify these potential candidate items.

\begin{comment}

\begin{theorem}
In the bounded deletion model, where $D\leq(1-\frac{1}{\alpha})I$, there exists an algorithm that can report all item with frequency higher than $\epsilon(I-D)$ after processing total $I$ insertions using $O(\frac{\alpha}{\epsilon})$ space.
\end{theorem}
\end{comment}

\textcolor{black}{By Lemma~\ref{important lemma1} and Lemma~\ref{important lemma2},} SpaceSaving~\cite{metwally2005efficient} with space $k$ reports all the items with frequency greater than or equal to $\frac{I}{k}$. Therefore by using $k=\frac{\alpha}{\epsilon}$ space to process $I$ insertions on the SpaceSaving algorithm, it will report all item with frequency greater than or equal to $\frac{\epsilon}{\alpha}I$. Since we know $\frac{1}{\alpha} \leq \frac{(I-D)}{I}$, $\frac{\epsilon}{\alpha}I \leq \epsilon I \frac{(I-D)}{I} = \epsilon(I-D)$. Hence by using $\frac{\alpha}{\epsilon}$ counters, all the items with frequency greater than or equal to $\epsilon(I-D)$ will be identified.

Interestingly, we find that modifying the original SpaceSaving algorithm with $O(\frac{\alpha}{\epsilon})$ space leads to an algorithm that solves the frequency estimation and frequent items problems in the bounded deletion model. The Lazy SpaceSaving$^\pm$ algorithm handles insertions exactly in the same manner as in the original Algorithm~\ref{Space Saving algorithm}. For deletions, the Lazy SpaceSaving$^\pm$ decreases the monitored item counter, if the deleted item is monitored. Otherwise, the deletions on unmonitored item are ignored, as shown in Algorithm~\ref{Lazy Space Saving}. \textcolor{black}{The frequency is still estimated according to Algorithm~\ref{Space Saving query}.} The rationale for this design is that an unmonitored item has estimated frequency of 0 and deletions of the unmonitored items will not amplify the difference but in fact narrows the difference. Initially, this may seem to be counter-intuitive. Another way to think about it is that the frequency estimations of the unmonitored items can only be underestimates. Thus, the decrease in an unmonitored item's true frequency narrows the underestimation.
\begin{algorithm}[]
\SetAlgoLined
    %$Sketch \leftarrow  (\emptyset$, $\frac{\alpha}{\epsilon}$)\; 
    %// $\frac{\alpha}{\epsilon}$ is the sketch size\;
    \For{item from deletions}{
        \uIf{item in Sketch}{
            $count_{item}$ -= 1 \;
        }
        \uElse{
            //ignore
        }
    }
 \caption{Lazy SpaceSaving$^\pm$: Deletion Handling}
 \label{Lazy Space Saving}
\end{algorithm}

We now formally establish that Algorithm~\ref{Lazy Space Saving} solves the frequency estimation problem in the bounded deletion model. \textcolor{black}{Let $error_{max}$ be the maximum frequency estimation error based on the sketch. We show by induction that $error_{max}$ is always less than $\epsilon(I-D)$.}
\begin{theorem} \label{theorem: lazy frequency estimation}
In the bounded deletion model where $D\leq(1-\frac{1}{\alpha})I$, after processing $I$ insertions and D $deletions$, Lazy SpaceSaving$^{\pm}$ using $O(\frac{\alpha}{\epsilon})$ space solves the frequency estimation problem in which $\forall i, |f(i) - \hat{f}(i)| < \epsilon(I-D)$ where $f(i)$ and $\hat{f}(i)$ are the exact and estimated frequencies of an item $i$.
\end{theorem}

\begin{proof} By Induction.

\textit{Base case:} After $i < I$ insertions and 0 deletions with $O(\frac{\alpha}{\epsilon})$ space, $error_{max}$ is less than $\epsilon(I-D)$. Hence, by Lemma~\ref{lemma error bound} (of the original insertion-only $SpaceSaving$), $error_{max} < i\frac{\epsilon}{\alpha} \leq \epsilon \frac{i(I-D)}{I} < \epsilon(I-D)$.

\textit{Induction hypothesis:} After $i < I$ insertions and $d < D$ deletions, the maximum frequency estimation error of the sketch is $error_{max} < \epsilon(I-D)$.

\textit{Induction Step:} Consider the case when the $(i+d+1)^{th}$ input item is an insertion. If the newly inserted item $x$ is monitored or the sketch is not full, then no error is introduced. If the newly inserted item $x$ is not monitored and the sketch is full, then $x$ replaces the $minItem$ which is the item with minimum count, $minCount$, in all monitored items. \textcolor{black}{$minCount$ is maximized when every item inside the sketch has the same count, and hence $minCount \leq i\frac{\epsilon}{\alpha}<\epsilon(I-D)$. The estimated frequency for $x$ is $minCount$+1 and $x$ is at most overestimated by $minCount$. The frequency estimation for $minItem$ becomes 0, and $minItem$'s frequency estimation is off by at most $minCount$.} Therefore, $error_{max}$ after processing the newly inserted item is still less than $\epsilon(I-D)$. 

 Consider the case when the $(i+d+1)^{th}$ input item is a deletion. If the newly deleted item $x$ is monitored, its corresponding counter will be decremented and no extra error is introduced and $error_{max}$ is still less than $\epsilon(I-D)$.
 %$0 \leq \hat{f}(x)-f(x) = (\hat{f}(x)-1)-(f(x)-1) \leq error_{max} < \epsilon(I-D)$. 
 If the newly deleted item $x$ is not monitored, then the algorithm ignores this deletion. \textcolor{black}{The frequency estimation errors for monitored items do not change and they are still less than $\epsilon(I-D)$.} Moreover, before the arrival of $x$, $\forall i \notin Sketch, f(i)-\hat{f}(i) = f(i) - 0 < \epsilon(I-D)$. By ignoring the deletion of the unmonitored items, $\forall i \notin Sketch, (f(i)-1)-\hat{f}(i) < f(i) - \hat{f}(i) < \epsilon(I-D)$. 

\textit{Conclusion:} By the principle of induction, Lazy SpaceSaving$^\pm$ using $O(\frac{\alpha}{\epsilon})$ space solves the frequency estimation problem with bounded error, i.e, $\forall i, |f(i) - \hat{f}(i)| < \epsilon(I-D)$.
\end{proof}

%\subsection{Lazy SpaceSaving$^\pm$ Frequent Items}
\textcolor{black}{Lazy SpaceSaving$^\pm$ also solves the frequent items problem. To prove this, we first show Lazy SpaceSaving$^\pm$ never underestimates the frequency of a monitored item.} 

\begin{lemma} \label{lemma: lazy nvr underestimate}
\textcolor{black}{Lazy SpaceSaving$^\pm$ never underestimates the frequency of monitored items.}
\end{lemma}
\begin{proof}
\textcolor{black}{Since the handling of insertions is the same as the SpaceSaving and SpaceSaving never underestimates the frequency of monitored items by Lemma~\ref{lemma: spacesaving overestimation}, it is clear that the insertions can not lead to frequency underestimation for monitored items. When handling deletions, Lazy SpaceSaving$^\pm$ only decrements the count when the deleted item is monitored. Since the deletion of a monitored item implies its true frequency and its estimated frequency both decrease by one, this procedure has no effect on the frequency estimation error. Therefore, Lazy SpaceSaving$^\pm$ never underestimates the frequency of monitored items.}
\end{proof}

%It overestimate frequencies when either a monitored item is evicted after an insertion, or a deletion of an unmonitored item is ignored. 

Since Lazy SpaceSaving$^\pm$ never underestimates, then if we report all the items with frequency estimations greater or equal to $\epsilon (I-D)$, then all frequent items will be reported. %This can be proved by contradiction.
 
\begin{theorem}
In the bounded deletion model, where $D\leq(1-\frac{1}{\alpha})I$, Lazy SpaceSaving$^{\pm}$ solves the frequent items problem using $O(\frac{\alpha}{\epsilon})$ space.
\end{theorem}

\begin{proof} By Contradiction.

 \textcolor{black}{Assume a frequent item $x$ is not reported and by definition of frequent items, $f(x) \geq \epsilon (I-D)$. Since it is not reported, its frequency estimation, $\hat{f}(x)$, must be less than $\epsilon (I-D)$. There are two cases where $x$ will not be reported: (i) $x$ is not monitored, or (ii) $x$ is monitored, but its frequency is underestimated, i.e., $\hat{f}(x)<\epsilon(I-D)$.}
 
 \textcolor{black}{In the first case where $x$ is not monitored, the estimation frequency of $x$ is 0, i.e, $\hat{f}(x)=0$. Since $x$ is by assumption a frequent item, the frequency estimation difference for $x$ is $|\hat{f}(x) - f(x)| \geq  \epsilon (I-D)$. However, this contradicts Theorem~\ref{theorem: lazy frequency estimation} in which any items' frequency estimation error is strictly less than $\epsilon(I-D)$.}
 
  \textcolor{black}{In the second case where $x$ is monitored but not reported, its estimated frequency is less than $\epsilon(I-D)$, i.e., the frequency estimation for item $x$ is an underestimation. However, by Lemma~\ref{lemma: lazy nvr underestimate}, Lazy SpaceSaving$^\pm$ never underestimates the frequency of monitored items.}
  
  \textcolor{black}{Hence, by contradiction Lazy SpaceSaving$^\pm$ solves the deterministic frequent items problem.}
\end{proof}

\subsection{An illustration of Lazy SpaceSaving$^\pm$} \label{lazy example}

\begin{figure}[tbph]
\centering
\includegraphics[scale=0.2]{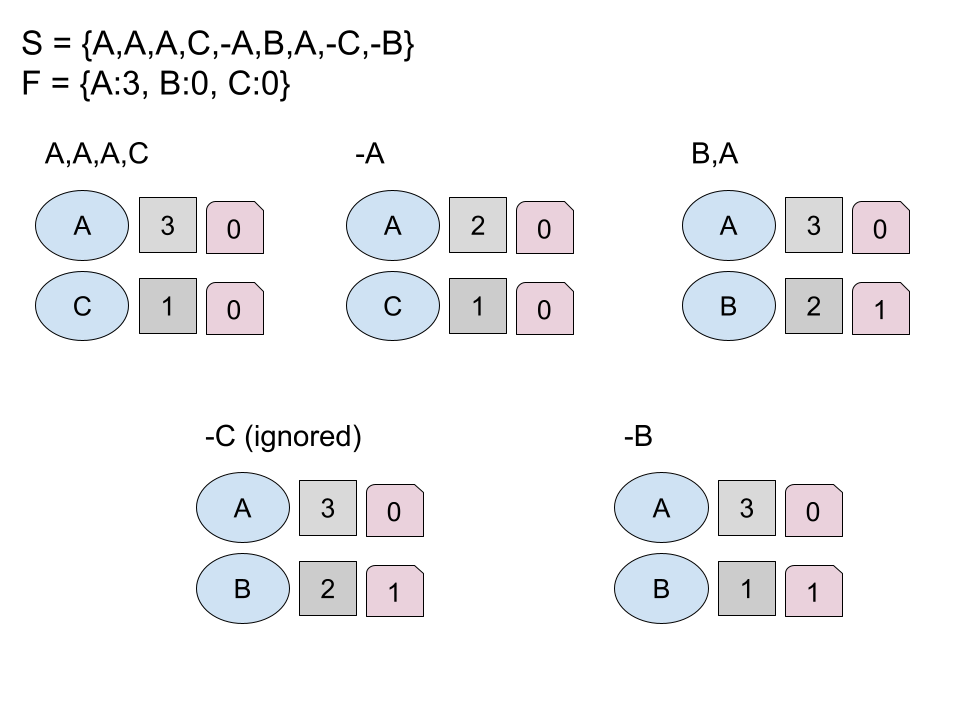}
\caption{Input Stream consisting of 6 insertions and 3 deletions. Each tuple represents an item, estimated frequency, and estimation error.}
\label{fig:lazy approach example}
\end{figure}

\textcolor{black}{Consider an instance of Lazy SpaceSaving$^\pm$ with capacity of 2. The input stream $\sigma$ is $(A,A,A,C,-A,B,A,-C,-B)$ where the minus sign indicate a deletion. The corresponding true frequency of $A$ is 3 while the true frequency of all other items is 0. For the first four insertions and one deletion of the monitored item $A$, the sketch maintains the exact count with no errors. When the sixth item $B$ arrives, $B$ replaces item $C$, since $C$ has the minimum count. The following insertion is $A$ and since $A$ is monitored, $A$'s count increases. Then items $-C,-B$ arrive. Since $C$ is not monitored, Lazy SpaceSaving$^\pm$ ignores the deletion of $C$, and the deletion of monitored item $B$ decreases the corresponding count, as shown in Figure~\ref{fig:lazy approach example}. After processing all inputs, the lazy-approach does not underestimate the frequency of the items in the sketch (it overestimates the frequency of item $B$). The maximum frequency estimation error is 1 since $\hat{f}(A)-f(A) = 0$ and $\hat{f}(B)-f(B) = 1$.}  

%Note that this example illustrates that by using two counters, Lazy SpaceSaving$^\pm$ correctly reports the most heavy item in the bounded deletion model with delete:insert ratio at most 0.5 in one pass. We know that for the insertion-only model, the counter based solutions~\cite{boyer1991mjrty, metwally2005efficient} reports the most heavy item using one counter. In general, for an input bounded deletion stream with $\alpha \geq 1$, Lazy SpaceSaving$^\pm$ algorithm increases the number of counters by a factor of $\alpha$ to tolerate at most $1-\frac{1}{\alpha}$ deletions and is still guaranteed to report the $\phi$-frequent items.

\subsection{SpaceSaving$^\pm$}
\textcolor{black}{While Lazy SpaceSaving$^\pm$ elegantly satisfies all the necessary requirements, the average estimation error and total estimation error may increase if there are significant deletions of the unmonitored items. Therefore, we propose SpaceSaving$^\pm$, a novel algorithm and a data structure that efficiently handles deletions of the unmonitored items. Interestingly, we experimentally show that SpaceSaving$^\pm$ performs better than Lazy SpaceSaving$^\pm$ when they are both allocated the same sketch space, even though we need more space by a constant factor to establish the correctness of SpaceSaving$^\pm$.}

Both the original SpaceSaving and our proposed Lazy SpaceSaving$^\pm$ algorithms have the property of never underestimating the frequency of the monitored items. Since the $\epsilon$-approximation requirement is $\forall i, |f(i) - \hat{f}(i)| < \epsilon(I-D)$, \textcolor{black}{there are opportunities to reduce the amount of overestimation for the monitored items, as long as the frequency estimation error is still within this bound. We observe that an item with a large estimation error indicates that it is unlikely to be a heavy item, as heavy items are often never evicted from the sketch and have small estimation error. In addition, items with large estimation error are often overestimated due to the aggregation of the frequencies of many less-weighted items.} SpaceSaving$^\pm$ leverages this intuition. It handles the insertions of all items, and the deletions of the monitored items exactly in the same way as the Lazy SpaceSaving$^\pm$. For the deletions of the unmonitored items, SpaceSaving$^\pm$ decrements the count associated with the item that has the maximum estimation error inside the sketch, as shown in Algorithm~\ref{bounded deletion Space Saving}. \textcolor{black}{With this modification, With this modification, the estimated frequency of any item reduces either from being replaced or from a deletion of an unmonitored item. In the following proofs, SpaceSaving$^\pm$ uses $\frac{2\alpha}{\epsilon}$ to ensure (i) no item can be severely overestimated by SpaceSaving$^\pm$, and (ii) no item can be severely underestimated by SpaceSaving$^\pm$. To estimate the frequency of an item, we still use Algorithm~\ref{Space Saving query}. Before analyzing the correctness of the algorithm, we first construct three very helpful lemmas.}

\begin{lemma} \label{lemma: upperbound on minCount in Bounded Deletion Model}
\textcolor{black}{The minimum count, $minCount$, in SpaceSaving$^\pm$ with $\frac{2\alpha}{\epsilon}$ counters is less than or equal to $\frac{\epsilon}{2}(I-D)$.}
\end{lemma}
\begin{proof}
\textcolor{black}{Since deletions never increment any counts, $minCount$ is maximized by processing $I$ insertions with no deletions. With $I$ insertions and no deletions, the sum of all the counts is equal to $I$. The $minCount$ is the largest when all the other counts are the same as $minCount$. Hence, $minCount \leq \frac{\epsilon I}{2 \alpha} \leq \frac{\epsilon(I-D)}{2}$.}
\end{proof}

\begin{lemma} \label{lemma: max estimation error upper bound}
\label{lemma: SpaceSaving+- estimation error upper bound}
\textcolor{black}{The maximum estimation error in SpaceSaving$^\pm$ with $\frac{2\alpha}{\epsilon}$ counters is less than $\frac{\epsilon}{2}(I-D).$}
\end{lemma}
\begin{proof}
\textcolor{black}{The estimation error only increase when $minItem$ is replaced by a newly inserted item and after the replacement, the estimation error becomes $minCount$.  $minCount$ is maximized when the input contains $I$ insertions and no deletions. Hence by Lemma~\ref{important lemma1}, SpaceSaving$^\pm$ with $\frac{2 \alpha}{\epsilon}$ counters has $minCount < \frac{\epsilon}{2}(I-D)$. The estimation error is at most $minCount$ and thus less than $\frac{\epsilon}{2}(I-D).$}
\end{proof}

\begin{lemma} \label{lemma: max estimation error lower bound} 
\label{lemma: SpaceSaving+- sum estiamted error}
\textcolor{black}{The sum of all estimation errors in SpaceSaving$^\pm$, is an upper bound on the sum of frequencies of all unmonitored items and the maximum estimation error is lower bounded by 0.
%i.e., $0 \leq \sum_{e \in unmonitored} f(e) \leq \sum_{i=1}^{k} error(item_i)$. 
%Since the sum of estimation error is upper bounded by $k$ times the maximum estimation error, the maximum estimation error is lower bounded by 0.
}
\end{lemma}
\begin{proof}
\textcolor{black}{The deletion of a monitored item has no effect on the sum of the estimation errors, and it has no effect on the sum of the frequencies of the unmonitored items. The deletion of an unmonitored item decreases both the sum of the frequencies of the unmonitored items by 1 and the sum of the estimation error by 1. From this observation and Lemma~\ref{lemma: SS estimation error upper}, we can conclude that in SpaceSaving$^\pm$ with $k$ counters, the sum of all estimation errors is an upper bound on the sum of frequencies of all unmonitored items. Since the sum of frequencies of all unmonitored items is above or equal to 0 and the sum of all estimation errors is upper bounded by $k$ times the maximum estimation error, the maximum estimation error is lower bounded by 0.}
\end{proof}

\begin{algorithm}[]
\SetAlgoLined
    %$Sketch \leftarrow  (\emptyset$, $\frac{\alpha}{\epsilon}$)\; 
    
    %// $\frac{\alpha}{\epsilon}$ is the sketch size
    
    \For{item from deletions}{
        \uIf{item in Sketch}{
            $count_{item}$ -= 1 \;
        }
        \uElse{
            j = arg $max_{j \in Sketch} error_{j}$ \;
            $count_{j}$ -= 1 \;
            $error_{j}$ -= 1 \;
        }
    }
 \caption{SpaceSaving$^\pm$: Deletion Handling}
 \label{bounded deletion Space Saving}
\end{algorithm}

 %Consider an instance of SpaceSaving$^\pm$ with space $\frac{2 \alpha}{\epsilon}$ to process $I$ insertions and $D$ deletions.

\begin{theorem} \label{theorem SpaceSaving+- freq}
In the bounded deletion model where $D\leq(1-1/\alpha)I$, after processing $I$ insertions and D $deletions$, SpaceSaving$^\pm$ using $O(\frac{\alpha}{\epsilon})$ space solves the frequency estimation problem in which $\forall i, |f(i) - \hat{f}(i)| < \epsilon(I-D)$ where $f(i)$ and $\hat{f}(i)$ are the exact and estimated frequencies of an item $i$.
\end{theorem} \label{SS+- frequency proof}

%Recall that the lazy approach never underestimates frequency of items inside the sketch. More precisely, the deletions of monitored items does not affect the error bound and the Lazy SpaceSaving$^\pm$ algorithm guarantees:
%\begin{equation*}
%\forall x \in Sketch, 0 \leq \hat{f}(x) - f(x) \leq error'(i) < \epsilon (I-D)
%\end{equation*}

\begin{proof}

\textcolor{black}{Consider an instance of SpaceSaving$^\pm$ with $\frac{2\alpha}{\epsilon}$ counters to process $I$ insertions and $D$ deletions. First, we prove there is no item $i$ such that the frequency estimate of $i$ severely overestimate its true frequency, i.e, $\nexists i, \hat{f}(i) - f(i) > \epsilon(I-D)$. In SpaceSaving$^\pm$, the handling of deletions can not lead to any overestimation as counters will only be decremented, and only the replacement of the $minItem$ due to a newly inserted item can lead to frequency overestimation of the newly inserted item. From lemma~\ref{lemma: upperbound on minCount in Bounded Deletion Model}, the $minCount$ in SpaceSaving$^\pm$ with $\frac{2\alpha}{\epsilon}$ counters is no more than $\frac{\epsilon}{2}(I-D)$. The overestimation of a newly inserted item can be at most $minCount$. Therefore, no item can be overestimated by more than $\frac{\epsilon}{2}(I-D)$.}

\textcolor{black}{Second, we prove there is no item that can be severely underestimated i.e, $\nexists i, \hat{f}(i) - f(i) < -\epsilon(I-D)$. Two operations may lead to underestimation: (i) Replacing $minItem$ can lead to frequency underestimation of $minItem$; (ii) Deletion of an unmonitored item can lead to frequency estimation of the item with maximum estimation error. For the first case, $minCount$ is always less than $\frac{\epsilon}{2}(I-D)$, and the amount of underestimation is less than $\frac{\epsilon}{2}(I-D)$ for any item due to the replacement.}

\textcolor{black}{We show that the deletion of an unmonitored item can lead to at most $\frac{\epsilon}{2}(I-D)$ frequency underestimation. Based on Lemma~\ref{lemma: max estimation error upper bound} and Lemma~\ref{lemma: max estimation error lower bound}, the maximum estimation error must be less than $\frac{\epsilon}{2}(I-D)$ and larger or equal to 0.
In Algorithm~\ref{bounded deletion Space Saving}, lines 6 and 7, the deletion of an unmonitored item decreases both the count and the estimation error of the item with the maximum estimation error. Call this item $x$. Since $x$'s counter decreases by 1, the difference between $x$'s frequency estimation and $x$'s true frequency, $\hat{f}(x) - f(x)$, also decreases by 1. Since the maximum estimation error is between $\frac{\epsilon}{2}(I-D)$ and 0, the number of decrements due to an unmonitored item is at most $\frac{\epsilon}{2}(I-D)$ times. Hence for any item, its frequency error is underestimated by at most $\frac{\epsilon}{2}(I-D)$ due to the deletion of unmonitored item. As a result, for any item, its frequency can be underestimated at most by $\epsilon(I-D)$ from replacing the $minItem$ and the deletions of the unmonitored items.}
\end{proof}

In Theorem~\ref{theorem SpaceSaving+- freq}, we proved that SpaceSaving$^\pm$ guarantees that all frequency estimations are off by no more than $\epsilon (I-D)$, i.e., $\forall i, |\hat{f}(i)-f(i)| \leq \epsilon(I-D)$. Note that unlike the original SpaceSaving, SpaceSaving$^\pm$ may underestimate but never severely underestimates. By reporting all the items with estimated frequency larger than 0, all frequent items must be reported, which can be proved by contradiction. 

\begin{theorem} \label{theorem spavesaving frequent}
In the bounded deletion model, where $D\leq(1-\frac{1}{\alpha})I$, SpaceSaving$^{\pm}$ solves the frequent items problem using $O(\frac{\alpha}{\epsilon})$ space.
\end{theorem}

\begin{proof} Proof by contradiction:

Assume SpaceSaving$^{\pm}$ algorithm does not report all frequent items. Then, there must exists a frequent item $x$ that is not reported. Since SpaceSaving$^{\pm}$ reports all the items with estimation frequency larger than 0 as frequent items (recall unmonitored items have estimated frequency of 0), $x$'s estimated frequency must be less than or equal to 0, i.e., $\hat{f}(x) \leq 0$. Moreover, since $x$ is a frequent item, then the true frequency estimation of $x$ must be larger than $\epsilon(I-D)$, i.e., $f(x) > \epsilon(I-D)$. The difference between the estimated frequency and its true frequency is then off by more than $\epsilon (I-D)$, i.e., $|f(x)-\hat{f}(x)| > \epsilon(I-D)$. This leads to a contradiction since it violates the $\epsilon$-approximation guarantee proved in Theorem~\ref{theorem SpaceSaving+- freq}.
\end{proof}
\subsection{An illustration of SpaceSaving$^\pm$}

\begin{figure}[tbph]
\centering
\includegraphics[scale=0.2]{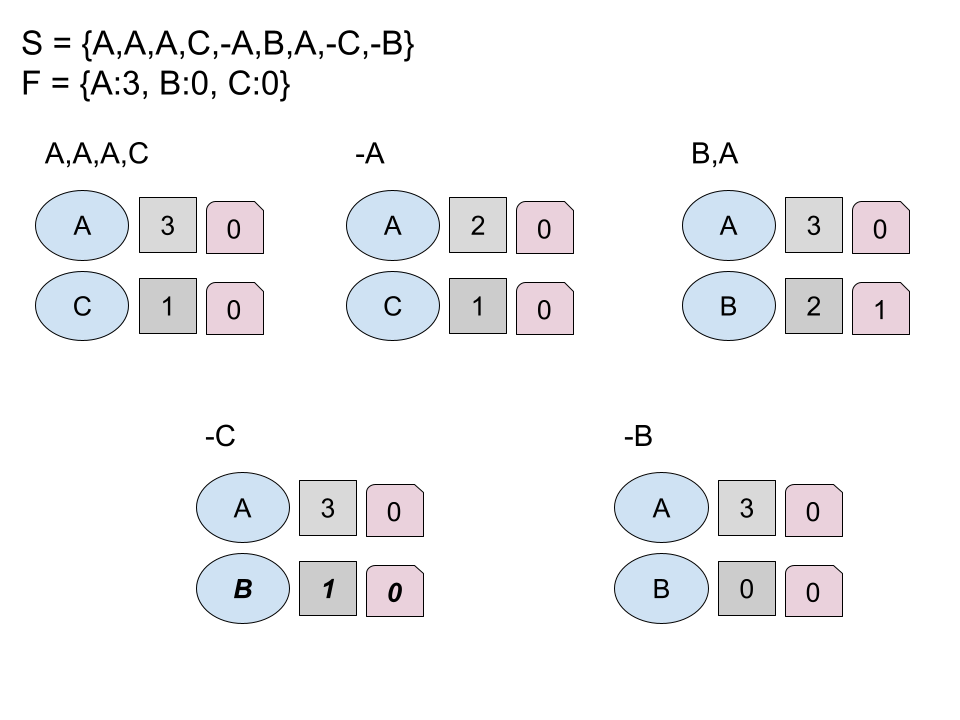}
\caption{Input Stream consists of 6 insertions and 3 deletions. Each tuple represents an item, estimated frequency, and estimation error.}
\label{fig:spacesaving +- approach example}
\end{figure}

\textcolor{black}{Consider the same input stream illustrated in Section~\ref{lazy example} and create an instance of SpaceSaving$^\pm$ with capacity of 2, in which input stream $\sigma$ is $(A,A,A,C,-A,B,A,-C,-B)$ where the minus sign indicate a deletion. The corresponding true frequency of $A$ is 3, while the true frequency of all other items is 0. The sketch image after digesting the first 7 items are exactly the same as in the previous example. When the deletion of item $C$ comes in, SpaceSaving$^\pm$ does not ignore the deletion of unmonitored item $C$, and since item $B$ has the largest estimation error, both $B$'s count and $B$'s estimation error are decreased. The final deletion of $A$ decreased $A$'s corresponding count. After processing the stream, the estimated frequency for $A$ and $B$ are 3 and 0 respectively, as shown in Figure~\ref{fig:spacesaving +- approach example}. The maximum frequency estimation error is 0 since $|\hat{f}(A)-f(A)| = 0$ and $|\hat{f}(B)-f(B)| = 0$. With the same bounded deletion stream and sketch space, Lazy SpaceSaving$^\pm$ overestimated the frequency of item $B$ by 1 (Section~\ref{lazy example}) and SpaceSaving$^\pm$ is able to further reduce the estimation error to 0. By judiciously handling the deletion of the unmonitored items, SpaceSaving$^\pm$ reduces the impact of overestimation and achieves better accuracy. }

\subsection{Min Heap and Max Heap}
\label{min-max-heap}
%The original SpaceSaving algorithm~\cite{metwally2005efficient} was proposed for click streams with unit weight updates and uses a linked list data structure to dynamically summarize the insertion stream. Later, Berinde et al.~\cite{berinde2010space} generalized the problem context to weighted updates. To support weighted updates efficiently, the S
SpaceSaving algorithm is usually implemented with a standard min-heap data structure such that the operations that increase the item weights and that remove the minimum item can be performed in logarithmic time~\cite{berinde2010space}. To support the deletion of the unmonitored items, the SpaceSaving$^\pm$ algorithm further needs to find the item with the largest estimation error \textcolor{black}{ and modify the estimation errors efficiently.} From these observations, we use two heaps on both the estimated counts and the estimation errors, as underlying data structures. The estimated counts are stored in a min heap,  the estimation errors are stored in a max heap, \textcolor{black}{and a dictionary maps each item to the corresponding nodes in these two heaps,} as shown in Figure~\ref{fig:min-max-heap}. \textcolor{black}{Using two heaps and a dictionary with $O(k)$ space, both the minimum count and maximum estimation error can be found in $O(1)$ time; while insertions and deletions can be done in $O(\log(k))$ time. For example, if the sketch needs to delete an unmonitored item, then the procedures are: (1) use the dictionary to find the deletion is performed on an unmonitored item; (2) use the max heap to find the item with maximum estimation error; (3) use the dictionary to find the location of this item; (4) decrease both its count and its estimation error; (5) percolate down in max heap and percolate up in min heap;}

\begin{figure}[tbph]
\centering
\includegraphics[scale=0.2]{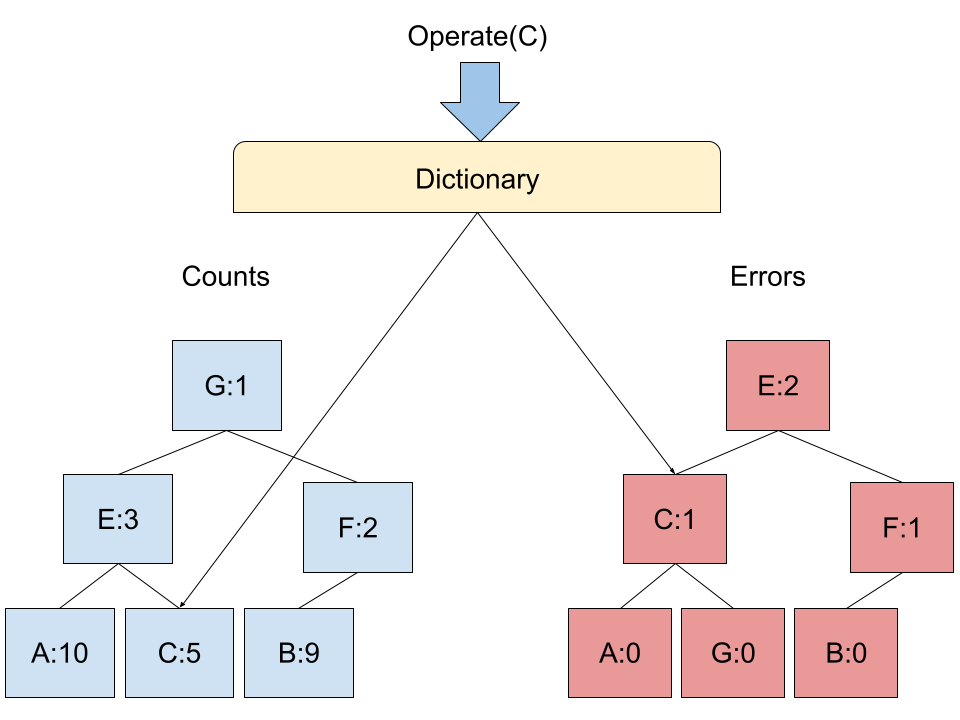}
\caption{\textcolor{black}{A dictionary points to the left min heap with items' count, and right max heap with items' estimation errors}.}
\label{fig:min-max-heap}
\end{figure}

%% file: DSS.tex
\section{Quantile Sketch}
\label{sec-dydadic}
In this section, we propose Dyadic SpaceSaving$^\pm$ sketch, the first deterministic quantile sketch in the bounded deletion model. The Dyadic SpaceSaving$^\pm$ sketch is a universe-driven algorithm that accurately approximates quantiles with strong guarantees. 

\subsection{The Quantiles Problem}

The rank of an element $x$ is the total number of elements that are less than or equal to $x$, denoted as $R(x)$. The quantile of an element $x$ is defined as $R(x)/|F|_1$ where $F$ is the frequency vector. 
The most familiar quantile value is $0.5$ also known as \textit{median}. \textit{Deterministic} $\epsilon$ approximation quantile algorithms~\cite{greenwald2001space, shrivastava2004medians} take as input a precision value $\epsilon$ and an item such that the approximated rank has at most $\epsilon |F|_1$ additive error. The \textit{randomized} quantile algorithms provide a weaker guarantee in which the approximated rank of an item has at most $\epsilon |F|_1$ additive error with high probability.~\cite{ivkin2019streaming, karnin2016optimal, manku1998approximate}. 
%These algorithms provide guarantees by bounding the failure probability to be at most $\delta$ such that the user has $1-\delta$ confidence that the sketch’s output is $\epsilon$ approximation. 

Recently, Zhao et al.~\cite{zhaokll} proposed the first randomized quantile sketch KLL$^\pm$ in the bounded deletion model by generalizing the KLL~\cite{karnin2016optimal} from the insertion-only model. The first data sketch to summarize quantiles in the turnstile model was proposed by Gilbert et al.~\cite{gilbert2002summarize}, which breaks down the universe into dyadic intervals and maintains frequency estimations of elements for each interval.
% using a total space of $O(\frac{1}{\epsilon^2}\log^{2}{U \log{\frac{\log{U}}{\epsilon}}})$ and update time $O(\frac{1}{\epsilon^2}log^{2}Ulog(\frac{logU}{\epsilon}))$. 
Later, Cormode et al.~\cite{cormode2005improved} proposed the Dyadic Count-Min (DCM) sketch which replaces the frequency estimation sketch for each dyadic interval with a Count-Min, improving the overall space complexity to $O(\frac{1}{\epsilon} \log^{2}{U} \log{(\frac{\log{U}}{\epsilon})}))$ and update time to $O( \log{U}\log{(\frac{\log{U}}{\epsilon})})$. Then, Wang et al.~\cite{wang2013quantiles} proposed the Dyadic Count-Median sketch which replaces the Count-Min with the Count-Median~\cite{charikar2002finding} to further improve the space complexity to $O(\frac{1}{\epsilon} \log^{1.5}{U \log^{1.5}{(\frac{\log{U}}{\epsilon})})})$, while using the same update time complexity as DCM sketches. 

%For more detailed information, ~\cite{cormode2020small} provides comprehensive and detailed background information.

\subsection{DSS$^\pm$: A Deterministic Quantile Sketch}
We propose the Dyadic SpaceSaving$^\pm$ sketch to solve deterministic quantile approximation in the bounded deletion model. Inspired by the previous algorithms, we observe that by replacing the frequency estimation sketch in each dyadic layer with a SpaceSaving$^\pm$ of space $O(\frac{\alpha}{\epsilon}log(U))$ solves the quantile approximation in the bounded deletion model. Any range can be decomposed into at most $logU$ dyadic intervals~\cite{cormode2019answering}. Since SpaceSaving$^\pm$ with $O(\frac{\alpha}{\epsilon}log(U))$ space ensures that the frequency estimation has at most $\frac{\epsilon(I-D)}{logU}$ additive error and by summing up at most $logU$ frequencies, the approximated rank has at most $\epsilon(I-D)$ additive error and the approximated quantile has at most $\epsilon$ error. To update the DSS$^\pm$ quantile sketch with an item $x$ of weight $w$: for each $logU$ layers, $x$ is mapped to an element in that layer and updates the corresponding element's frequency, as shown in Algorithm~\ref{Quantile Update}. The rank information of an item can be calculated by summing $O(logU)$ number of subset sums, as shown in Algorithm~\ref{Quantile Query}. Therefore, the Dyadic SpaceSaving$^\pm$ sketch requires $O(\frac{\alpha}{\epsilon}log^{2}(U))$ space with update time $O(logUlog\frac{\alpha logU}{\epsilon})$. %\textcolor{black}{Quantile experiments comparing DSS$^\pm$, KLL$^\pm$ and DCS are shown in Appendix~\ref{Quantile Expirement}.}

\begin{algorithm}[]
\SetAlgoLined
\SetKwInOut{Input}{Sketch}

    \For{h from 0 to logU}{
        DSS$^\pm$[h].update(x, w)\;
        
        x= x/2\;
    }
 \caption{DSS$^\pm$ Update(x,w)}
 \label{Quantile Update}
\end{algorithm}

\begin{algorithm}[]
\SetAlgoLined
\SetKwInOut{Input}{Sketch}
    Rank = 0\;
    \For{h from 0 to logU}{
        \uIf{ x is odd}{
            Rank = Rank + DSS$^\pm$[h].query(x)\;
        }
        x= x/2\;
    }
    return Rank\;
 \caption{DSS$^\pm$ Query(x)}
 \label{Quantile Query}
\end{algorithm}

%\subsection{An Illustration of DSS$^{\pm}$}
%Assume a bounded deletion data stream with $\alpha = 2$ where the items are drawn from a bounded universe as shown in Figure~\ref{fig:DSS}.The input stream is: 0,0,0,0,1,1,2,3,5,7,-0,-1,-2,-3,-5, where negative sign implies deletion. Take insert 6 as an example, the DSS$^\pm$ at level 0 will insert 6 (110), then DSS$^\pm$ at level 1 will insert 3 (11*) which represent the range of [6-7], etc., as shown in Algorithm~\ref{Quantile Update}. Moreover, to find the rank information of item 6, we can aggregate the frequencies of all items in range[0-6], which is equivalent to summing the frequencies of [0-3] at level 2, [4-5] at level 1, and [6] at level 0. Hence, the rank of item 6 is 5, as described in Algorithm~\ref{Quantile Query}. 

%\begin{figure}[tbph]
%\centering
%\includegraphics[scale=0.2]{figures/Dyadic DSS Quantile Sketch.png}
%\caption{Inputs drawn from bounded universe. Each layer of the dyadic intervals has a DSS$^{\pm}$ sketch of size 4.}
%Inputs drawn from bounded universe. 
%\label{fig:DSS}
%\end{figure}

%% file: experiments.tex
\section{Experiments}
\label{sec-eval}
This section evaluates the performance of Lazy SpaceSaving$^\pm$ and SpaceSaving$^\pm$. They are the first deterministic frequency estimation and frequent item algorithms in the bounded deletion model and make no assumptions on the universe. The experiments aim to identify advantages and disadvantages of lazy SpaceSaving$^\pm$ and SpaceSaving$^\pm$ compared to other state of the art sketches such as: 
\begin{itemize}
    \item  \textbf{CSSS}~\cite{jayaram2018data} : The CSSS sketch is the first theoretical algorithm to solve the frequency estimation and frequent item problems in the bounded deletion model.
    
    \item \textbf{Count-Min~\footnote{See https://github.com/rafacarrascosa/countminsketch for implementation detail}}~\cite{cormode2005improved}: The Count-Min Sketch operates in the turnstile model and the estimated frequencies are never underestimated.
    
    \item \textbf{Count-Median~\footnote{See~\cite{cormode2020small} for implementation detail}}~\cite{charikar2002finding}: The Count-Median Sketch operates in the turnstile model and its frequency estimation is unbiased.
\end{itemize}

The quantile evaluations in aims in identifying characteristic of Dyadic SpaceSaving$^\pm$ and compare it with other state of the art quantile sketches with deletion functionality:

\begin{itemize}
    \item \textbf{KLL$^\pm$}: \textcolor{black}{KLL$^\pm$~\cite{zhaokll} is the state of the art randomized quantile sketch that operates in the bounded deletion model, and it has no assumption on the universe.}
    
    \item \textbf{DCS}: \textcolor{black}{Dyadic Count Sketch~\cite{wang2013quantiles} is the state of the art randomized quantile sketch that operates in the turnstile model, and it leverage the bounded universe.} 
\end{itemize}

\subsection{Experimental Setup}

We implemented SpaceSaving$^\pm$ using the min and max heap data structure described in Section~\ref{min-max-heap} in Python 3.7.6.
The main distinction from the original SpaceSaving~\cite{metwally2005efficient} are: (i) the use of a min heap on weights and a max heap on the estimation errors;
(ii) support of delete operations using the lazy approach Algorithm~\ref{Lazy Space Saving} or SpaceSaving$^\pm$ Algorithm~\ref{bounded deletion Space Saving}; and (iii) the overall space complexity is $O(\alpha/\epsilon)$ and the update time complexity is $O(log(\alpha/\epsilon))$. We also implemented the CSSS sketch as described in~\cite{jayaram2018data}. All the experimental metrics are averaged over 5 independent runs. Moreover, in all experiments, Lazy SpaceSaving$^\pm$ and SpaceSaving$^\pm$ use the same amount of space, while the universe size is $U=2^{16}$, and we set $\delta = U^{-1}$ to align the experiments with the theoretical literature~\cite{bhattacharyya2018optimal, jayaram2018data}.

\subsection{Data Sets}
The experimental evaluation is conducted using both synthetic and real world data sets consisting of items that are inserted and deleted. For the synthetic data, we consider
three different distributions:
\begin{itemize}
    %\item \textbf{Uniform Distribution}: The insertions are randomly generated from a discrete uniform distribution, and the deletions are uniformly chosen from the insertions.
    
    \item \textbf{Zipf Distribution}: The elements are drawn from a bounded universe and the frequencies of elements follow the Zipf Law \cite{zipf2016human}, in which \textcolor{black}{ the frequency of an element with rank $R$: $f(R,s) = \frac{constant}{R^{s}}$ where $s$ indicates skewness.} Deletions are uniformly chosen from the insertions.

    \item \textbf{Binomial Distribution}: The elements are generated according to the binomial distribution with parameters $n$ and $p$ where $p$ is the probability of success in $n$ independent Bernoulli trials.
    
\end{itemize}

%In addition to the synthetic data sets, we used the following real world 2006 AOL Search Query Dataset~\footnote{The dataset can be found:http://www.cim.mcgill.ca/~dudek/206/Logs/AOL-user-ct-collection/}
\textcolor{black}{In addition to the synthetic data sets, we used the following real world CAIDA Anonymized Internet Trace 2015 Dataset~\cite{trace}.}
\begin{itemize}
\item \textbf{2015 CAIDA Dataset}:
%This is an extensive data set consisting of user search queries over three months in 2006. The dataset includes $<$\textit{AnonID, Query, QueryTime, ItemRank, ClickURL}$>$ tuples. In the experiments, deleted items are chosen from the inserted items and each update is the \textit{ClickURL}. The dataset has in total about 2 million valid urls.
\textcolor{black}{The CAIDA dataset is collected from the ‘equinixchicago’ high-speed monitor. In the experiment, we use 5 disjoint batches of 2 million TCP packets. For frequency estimation and frequent items evaluation, the insertions are the destination IP addresses and deletions are randomly chosen from insertions. For quantile evaluation, since DSS$^\pm$ need to perform division on the item, the insertions are the source port integer and deletions are randomly chosen from insertions.}
\end{itemize}

We also conducted experiments by exploring two additional patterns of the data sets:
\begin{itemize}

\item \textbf{Shuffled}: The insertions are randomly shuffled and the deletions are randomly and uniformly chosen from insertions.

\item \textbf{Targeted}: 
The insertions are randomly shuffled and the deletions delete the item with the least frequency.
\end{itemize}

\textcolor{black}{The metrics used in the experiments are averaged over 5 independent runs and they are:}
\begin{itemize}

\item \textbf{Mean Squared Error}: \textcolor{black}{The mean squared error (MSE) is the average of the squares of the frequency estimation errors.}

\item \textbf{Recall}: \textcolor{black}{The recall is defined as $\frac{TP}{TP+FN}$ where $TP$ (true positive) is the number of items that are estimated to be frequent and are indeed frequent and $FN$ (false negative) is the number of items that are frequent but not included in the estimations.}

\item \textbf{Precision}: \textcolor{black}{Precision is defined as $\frac{TP}{TP+FP}$ where $FP$ (false positive) is the number of items that are estimated to be frequent but are not frequent.}

\item \textbf{Kolmogorov-Smirnov divergence}: \textcolor{black}{Kolmogorov-Smirnov divergence \cite{cantelli1933sulla} is the \textit{maximum} deviation among all quantile queries, a measurement widely used to perform comparisons between two distributions~\cite{ivkin2019streaming}.}

\end{itemize}

\textcolor{black}{The experiments are presented in the following three subsections: frequency estimation, frequent item, and quantile approximation experiments.}

\subsection{Frequency Estimation Evaluation}

\begin{figure*}[h]
\centering
        \begin{subfigure}{.32\textwidth}
        \centering
        \includegraphics[scale = 0.4]{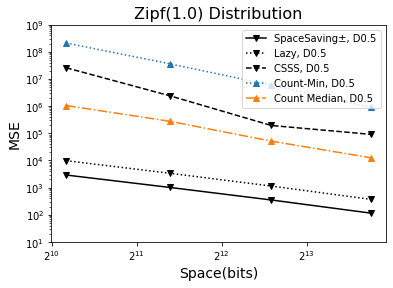}%
        \caption{}
        \end{subfigure}
        % \hfill
        \begin{subfigure}{.32\textwidth}
        \centering
        \includegraphics[scale = 0.4]{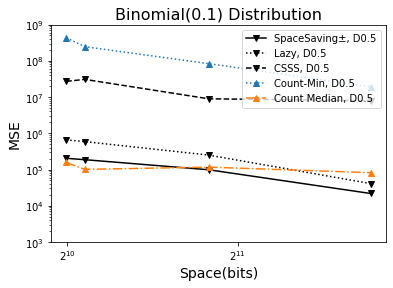}
        \caption{}
        \end{subfigure}
        % \hfill
        \begin{subfigure}{.32\textwidth}
        \centering
        \includegraphics[scale = 0.4]{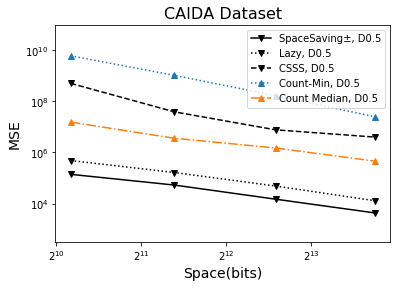}
        \caption{}
        \end{subfigure}
        
    \centering
        \begin{subfigure}{.32\textwidth}
        \centering
        \includegraphics[scale = 0.4]{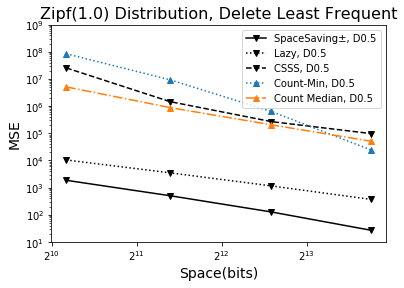}%
        \caption{}
        \end{subfigure}
        % \hfill
        \begin{subfigure}{.32\textwidth}
        \centering
        \includegraphics[scale = 0.4]{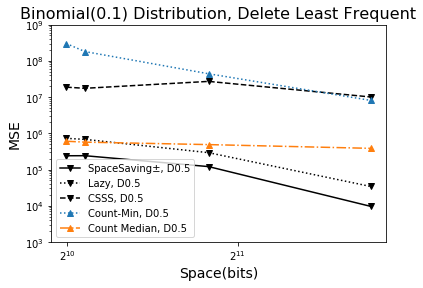}
        \caption{}
        \end{subfigure}
        % \hfill
        \begin{subfigure}{.32\textwidth}
        \centering
        \includegraphics[scale = 0.4]{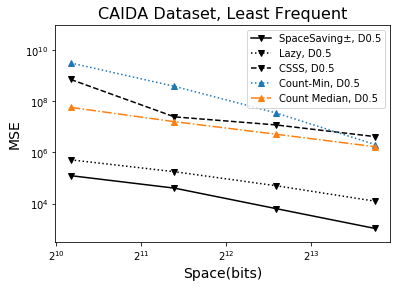}
        \caption{}
        \end{subfigure}
\caption{Trade-off between space and accuracy on various data distributions and different patterns}
\label{fig:MSE vs Space}
\end{figure*}

In this section, we compare Lazy SpaceSaving$^{\pm}$ and SpaceSaving$^{\pm}$ with state of the art Count-Min sketch, Count-Median Sketch, and CSSS Sketch. These experiments evaluate the accuracy of each sketch using the mean square error (MSE) while increasing the sketch size. The mean square error is the \textit{average} of the squares of the difference between items' estimated frequency and the true frequency, a measurement widely used to judge the accuracy of an estimation. It also serves as an empirical estimation of the variance~\cite{cormode2021frequency}. In MSE figures the x-axis denotes the sketch size while the y-axis depicts the average of the mean square errors. Since the mean square error is strictly positive, the lower y-axis values indicate better accuracy. In the following experiments, we assume all insertions arrive before any deletions into the sketch which is an adversarial pattern as spatial locality is minimized.
%and in later section, we explore the pattern of interleaved deletions.

\subsubsection{Sketch Size}

In this experiment, the input data has I insertions and D deletions, and the delete:insert ratio is 0.5. The deletion pattern is either shuffled, randomly chosen from insertions, or targeted delete of the least frequent items. The Zipf and Binomial distributions have $|F|_{1} = 10^5$ and \textcolor{black}{ the CAIDA dataset has $|F|_{1} = 10^6$, with two million insertions and one million deletions.} This experiment explores the effect of distribution skewness and the space size effect of sketches operating in both the bounded deletion model and in the turnstile model.

As expected, all sketches share the same pattern: increasing the sketch size leads to decrease in the MSE, shown in Figure~\ref{fig:MSE vs Space}. All experiments show SpaceSaving$^\pm$ has the lowest MSE and best accuracy as the sketch size grows. For the skewed Zipf distribution and CAIDA dataset, SpaceSaving$^\pm$ is the clear winner for all sketch sizes, as shown in Figure~\ref{fig:MSE vs Space}. For the lesser skewed binomial distribution, Count-Median performs competitively compared to SpaceSaving$^\pm$; however, SpaceSaving$^\pm$ eventually has better accuracy as the sketch size increases, as shown in Figure~\ref{fig:MSE vs Space}(b,e). The CSSS sketch has accuracy between Count-Median and Count-Min sketches. The Count-Min sketch often overestimates an item's frequency and thus has higher mean square error across all distribution. 

The targeted deletion pattern, when the least frequent items are targetted for deleteion, leads to a slight decrease in MSE across all distributions for Count-Min. The targeted delete pattern decreases the cardinality of $F$, increases the overall skewness, and hence heavy hitter items become more dominant and all sketches are able to capture the overall change and have less mean square error.

%As the delete-insert ratio increase, the accuracy drops for sketch in the bounded deletion model, as shown in Figure~\ref{fig:MSE vs Space}(a,b,d,e). In the AOL dataset, the SpaceSaving$^\pm$ with 0.5 delete-ratio has the lowest MSE. This phenomenon is due the not having fixed $I-D$. As the delete-insert ratio increase, the AOL dataset $I-D$ decreases, which implies the true frequency for each items also decreases and hence the difference between true frequency and estimated frequency also decreases.

\subsubsection{Delete:Insert Ratio}
\begin{figure}[tbph]
\centering
\includegraphics[scale=0.42]{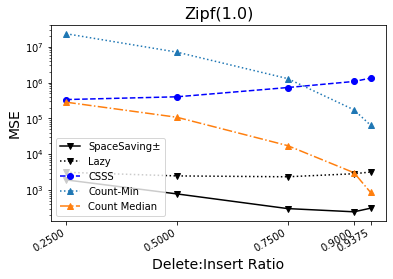}
\caption{Varying delete:insert ratio.}
\label{fig:ratio and error}
\end{figure}

Sketches in the bounded deletion model have their space complexity dependent on the parameter $\alpha$, which is an upper bounds on the delete:insert ratio. With higher delete:insert ratio, these sketches need to increase their sketch space to tolerate the increase in deletions in order to deliver the same guarantee. In this experiment, the sketch space is $10^3{3}logU$ bits and the input stream length is fixed to one million items. The x-axis represents different delete:insert ratio, and the y-axis is the mean squared error averaged over 5 independent runs, as shown in Figure~\ref{fig:ratio and error}. 

%The relative error is the absolute difference between estimated and true frequency divided by the true frequency $\frac{|\hat{f}x)-f(x)|}{f(x)}$, averaged over 5 independent runs. 

As expected, the accuracy of Lazy SpaceSaving$^\pm$ and CSSS depends on $\alpha$ and their MSE increases as the delete:insert ratio increases. The more interesting result is that SpaceSaving$^\pm$'s MSE decreases when the delete:insert ratio is less or equal to 0.75. Moreover, for a universe of size $2^{16}$, SpaceSaving$^\pm$ provides MSE less than CSSS, Count-Min, and Count-Median even if the delete:insert ratio is as high as 0.9375, which is $\frac{logU-1}{logU}$, while using the same amount of space, as shown by the right most dots in Figure~\ref{fig:ratio and error}. By handling the deletion of unmonitored items judiciously, SpaceSaving$^\pm$'s frequency estimation is more robust to the increase in deletions than other algorithms in the bounded deletion model. \textcolor{black}{For sketches that operate in the turnstile model, the MSE of Count-Min and Count-Median decreases as the delete:insert ratio increases, since more deletions reduce the number of hash collisions and reduce the amount of over counting in each bucket.} If the universe size increases, the performance of linear sketches will further decrease, whereas the data-driven SpaceSaving$^\pm$ has no dependency on the universe, and can provide accurate estimations even in the extreme case of unbounded universe.

\begin{figure}[tbph]
\centering
\includegraphics[scale=0.42]{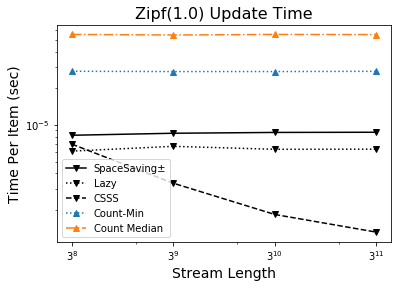}
\caption{Update times for Sketches}
\label{fig:updatetime}
\end{figure}

\begin{figure*}[h]
\centering
        \begin{subfigure}{.32\textwidth}
        \centering
        \includegraphics[scale = 0.36]{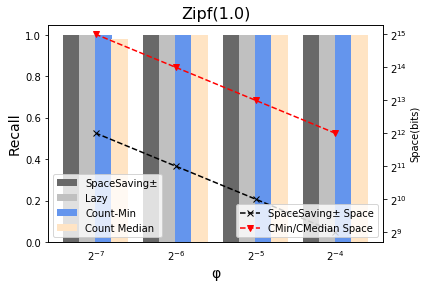}%
        \caption{}
        \end{subfigure}
        % \hfill
        \begin{subfigure}{.32\textwidth}
        \centering
        \includegraphics[scale = 0.36]{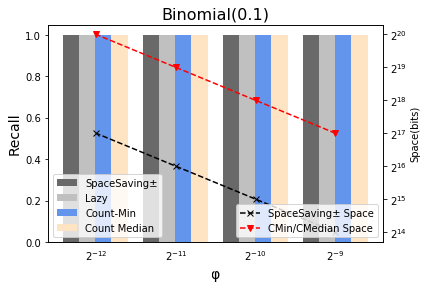}
        \caption{}
        \end{subfigure}
        % \hfill
        \begin{subfigure}{.32\textwidth}
        \centering
        \includegraphics[scale = 0.36]{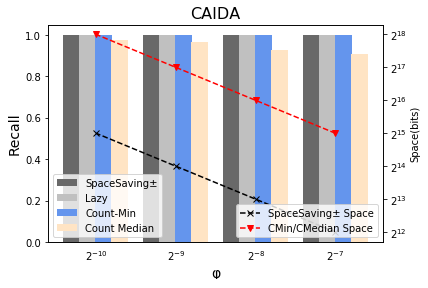}
        \caption{}
        \end{subfigure}
        
    \centering
        \begin{subfigure}{.32\textwidth}
        \centering
        \includegraphics[scale = 0.36]{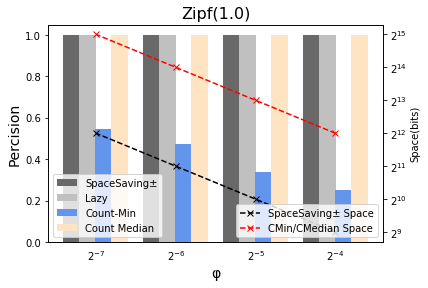}%
        \caption{}
        \end{subfigure}
        % \hfill
        \begin{subfigure}{.32\textwidth}
        \centering
        \includegraphics[scale = 0.36]{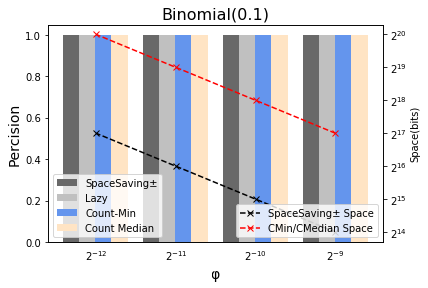}
        \caption{}
        \end{subfigure}
        % \hfill
        \begin{subfigure}{.32\textwidth}
        \centering
        \includegraphics[scale = 0.36]{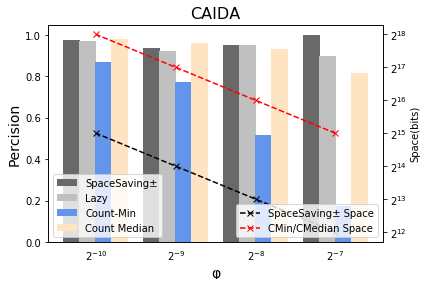}
        \caption{}
        \end{subfigure}
        
\caption{Recall and Precision Comparison} 
\label{fig:recall and precision}
\end{figure*}

\subsubsection{Update Time}

In Figure~\ref{fig:updatetime} \label{FE UpdateTime}, the x-axis is the stream length and the y-axis is the average latency in seconds per item over 5 independent runs. The input is a shuffled Zipf distribution and the delete:insert ratio is 0.5. All sketches use $10^3logU$ bits. As shown in Figure~\ref{fig:updatetime}, as expected, Lazy SpaceSaving$^\pm$ has slightly less update time than SpaceSaving$^\pm$. Since deletions are randomly drawn from the insertions, the heavy frequency items are more likely to appear in the deletions and hence fewer unmonitored items appear in the deletions. The lazy approach ignores these deletions of unmonitored items and hence has faster update time. 
%All these sketches have update times that are indifferent to the stream length, whereas, the 
CSSS sketch update time decreases as the stream length grows because it performs sampling to obtain $O(\frac{\alpha logU}{\epsilon})$ samples and runs Count-Median on these samples. As the stream length increases the sample size increases at a much slower pace and thus the average update time per item decreases. Count-Min and Count-Median have update times depend on the universe size where a larger universe size will further increase the update time. Since Count-Median performs more hashes than Count-Min, Count-Median requires more update time than Count-Min.

\subsection{Frequent Items Evaluation}

In this section, we compare the recall and precision of Lazy SpaceSaving$^\pm$ and SpaceSaving$^\pm$ with state of the art sketches on identifying the frequent items. All experiments in this section have delete:insert ratio of 0.5. The left y-axis depicts either the average recall or average precision over 5 independent runs: higher y-axis values indicate better recall or precision. %The recall is calculated by dividing the true positives by the sum of true positives and false negatives. The precision is calculated by dividing the true positives by the sum of true positives and false positives. 
The right y-axis denotes the space used for each sketch where Lazy SpaceSaving$^\pm$ and SpaceSaving$^\pm$ use $\frac{\alpha}{\epsilon}logU$ bits; Count-Min and Count-Median use $\frac{1}{\epsilon}log^2U$ bits. In the following experiments, all insertions arrive before any deletions. Since all true frequent items appear more than $\epsilon (I-D)$ times, each sketch queries all potential items and then reports all items with estimated frequency larger than $\epsilon (I-D)$ as the frequent items. %Note, although all frequent items are contained in SpaceSaving$^{\pm}$, it needs to report all items with frequency larger than 0 in order to guarantee 100\% recall, as proved in Theorem~\ref{theorem spavesaving frequent}. 
In addition, the following experiments do not compare with the CSSS sketch. Although CSSS can solve the frequent item problem, CSSS is more of theoretical interest since it reduces the size of each counter from $O(logU)$ bits to $O(log(\alpha))$ bits but in practice, it requires a lot more space to solve the frequent item problem. More specifically, the sketch size increases by $192$ times.
%which implies the universe is powered by 192 times, as $192\log U = \log U^{192}$. 
The space increase is more significant than the space saved by reducing the number of bits per counter. %Whereas, the linear sketch Count-Min and Count-Median sketch do not need to change $\epsilon$.

\subsubsection{Recall}
In these experiment, we compare the recall among Lazy SpaceSaveing$^\pm$, SpaceSaving$^\pm$, Count-Min and Count-Median. In Figure~\ref{fig:recall and precision} (a), (b), and (c), the x-axis represents different frequent items threshold $\phi$ in which frequent items have frequency larger than or equal to $\phi |F|_1$. The right y-axis denotes the space budget in which Lazy SpaceSaving$^\pm$ and SpaceSaving$^\pm$ use $\frac{\alpha}{\epsilon}$ space; Count-Min and Count-Median use $\frac{logU}{\epsilon}$ space. The sketch space increases as $\phi$ decreases. The left y-axis is the recall ratio. As expected, Lazy SpaceSaving$^\pm$ and Count-Min sketch have 100\% recall across all distributions, since they never underestimate the true frequency. The Count-Median sketch may sometimes underestimate the frequency and thus does not always achieve 100\% recall. Theorem~\ref{theorem spavesaving frequent} shows that SpaceSaving$^\pm$ needs to report all items with frequency larger than 0 to achieve 100\% recall. In this experiment, SpaceSaving$^\pm$ reports items with frequency larger than $\phi |F|_1$. Since it might underestimate an item's frequency, the recall rate might not be 100\%. However, in these experiments, SpaceSaving$^\pm$ still achieves 100\% recall across all distributions.

\begin{figure*}[h]
    \centering
        \begin{subfigure}{.32\textwidth}
        \centering
        \includegraphics[scale = 0.4]{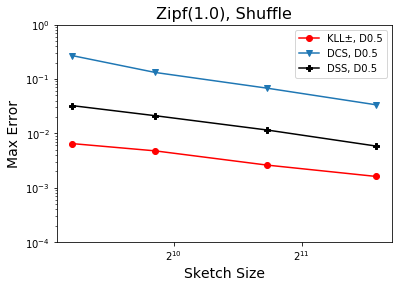}%
        \caption{}
        \end{subfigure}
        % \hfill
        \begin{subfigure}{.32\textwidth}
        \centering
        \includegraphics[scale = 0.4]{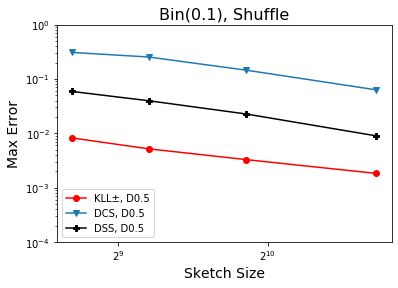}
        \caption{}
        \end{subfigure}
        % \hfill
        \begin{subfigure}{.32\textwidth}
        \centering
        \includegraphics[scale = 0.4]{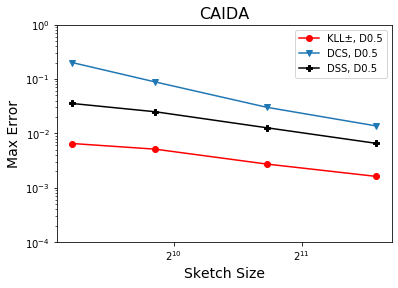}
        \caption{}
        \end{subfigure}
\caption{Quantile sketches trade-off between space and accuracy on different data distributions.}
\label{fig:Quantile Accuracy}
\end{figure*}

\subsubsection{Precision}
In this subsection, we compare the precision rates among Lazy SpaceSaveing$^\pm$, SpaceSaving$^\pm$, Count-Min and Count-Median. In Figure~\ref{fig:recall and precision} (d), (e), and (f), the x-axis represents the different frequent items threshold $\phi$. The right y-axis denotes the space budget in which Lazy SpaceSaving$^\pm$ and SpaceSaving$^\pm$ use $\frac{\alpha}{\epsilon}$ space; Count-Min and Count-Median use $\frac{logU}{\epsilon}$ space. The left y-axis is the precision ratio. Lazy SpaceSaving$^\pm$, SpaceSaving$^\pm$ and Count-Median have above 90\% precision for all $\phi$ and different distributions. SpaceSaving$^\pm$ has very high precision while using less space. Since Lazy SpaceSaving$^\pm$ sometimes overestimates an item's frequency by ignoring the deletion of unmonitored items, some items may be falsely classified as frequent items. Count-Min always overestimate items' frequencies and thus many items are incorrectly classified as frequent. %As the data distribution skewness increases, Count-Min's precision rate becomes higher.

\subsection{Quantile Evaluation} \label{Quantile Expirement}
\textcolor{black}{In this section, we experimentally evaluates Dyadic SpaceSaving$^\pm$, the first deterministic quantile approximation sketch in the bounded deletion model. We implemented the sketches in Python 3.7.6. Across all experiments, we assume the bounded universe size is $2^{16}$,  deletions are randomly chosen from insertions, and all insertions arrive before any deletions, an adversarial pattern that minimize locality.}

\subsubsection{Accuracy Comparison}

\textcolor{black}{In this section, we experimentally compares the quantile approximation accuracy among DSS$^\pm$, KLL$^\pm$, and DCS over different distributions and memory budgets. DSS$^\pm$ is deterministic, whereas KLL$^\pm$, and DCS are randomized. The binomial and zipf distribution has $|F|_1=10^5$, and CAIDA dataset has $|F|_1=10^6$. Across all distributions, sketches improve their accuracy as the memory budget increases, as shown in Figure~\ref{fig:Quantile Accuracy}. When the input change from binomial distribution to zipf distribution, the skewness increases and the DSS$^\pm$'s quantile approximation become more accurate, while the DCS's accuracy decreases. The reason is that, as the skewness increase, SpaceSaving$^\pm$'s frequency estimation becomes more accurate, and Count-Median's accuracy decrease which is also pointed out in~\cite{cormode2008finding}. KLL$^\pm$ performs the best and DSS$^\pm$ has better performance than DCS across all distributions. If the universe size becomes larger, than the performance of universe-based algorithm DSS$^\pm$ and DSS will further decrease.}

\subsubsection{Delete:Insert Ratio}

\begin{figure}[tbph]
\centering
\includegraphics[scale=0.42]{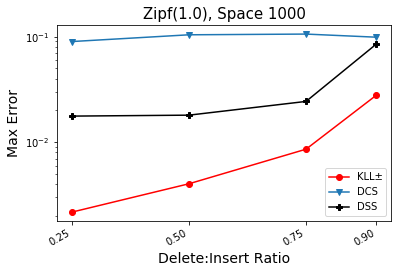}
\caption{Varying delete:insert ratio.}
\label{fig:Quantile ratio}
\end{figure}
\textcolor{black}{In this section, we experimentally compare DSS$^\pm$, KLL$^\pm$, and DCS sketches with different delete:insert ratio under same space budget of 1000. The input bounded deletion stream has length $10^6$ and it is shuffled zipf(1.0) distribution. In Figure~\ref{fig:Quantile ratio}, the y-axis is the maximum error and the x-axis is the delete:insert ratio and smaller y-value implies better accuracy. The result is aligned with the theoretical expectation, as the delete:insert ratio increase, the maximum error of quantile sketches operate in the bounded deletion model increases, since both KLL$^\pm$ and DSS$^\pm$ have dependence on $\alpha$ and $DCS$ has no dependence on $\alpha$. Using the same space budget, both KLL$^\pm$ and DSS$^\pm$ have better accuracy when the insert:delete ratio reaches 0.9. }

\subsubsection{Update Time}

\begin{figure}[tbph]
\centering
\includegraphics[scale=0.42]{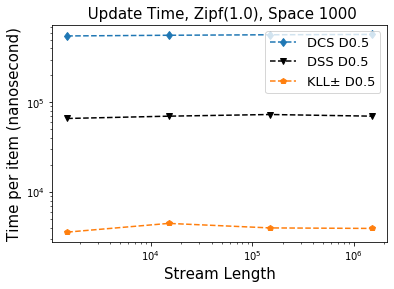}
\caption{Update times for Quantile Sketches}
\label{fig:Quantile Updatetime}
\end{figure}

\textcolor{black}{In this section, we experimentally compare the update time among DSS$^\pm$, KLL$^\pm$, and DCS sketches under same space budget. The input bounded deletion stream is shuffled zipf(1.0) distribution. In Figure~\ref{fig:Quantile Updatetime}, the y-axis is the update-time and the x-axis is the stream length and smaller y-value implies faster update time per item. The result is aligned with the theoretical expectation, the update time is independent from the stream length. Both DSS$^\pm$ and DCS rely on the dyadic structure over the bounded universe and have update time depend on the universe size. DSS$^\pm$ is faster than DCS as SpaceSaving$^\pm$ has faster update time than Count-Median, as shown section~\ref{FE UpdateTime}. KLL$^\pm$ uses sampling technique to digest arriving operations. It's update time is $O(log\frac{\alpha^{1.5}}{\epsilon})$ with no dependency on the universe size and hence it achieve the fastest update time.}

%% file: conclusion.tex
\section{Conclusion}
\label{sec-concl}
Frequency estimation and frequent items are two important problems in data stream research, and have significant impact for real world systems. Over the past decades of research, many algorithms have been proposed for the insertion-only and the turnstile models. In this work, we propose data-driven deterministic SpaceSaving$^\pm$ sketches, which maintain a subset of input items, to accurately approximate item frequency and report heavy hitter items in the bounded deletion model. To our knowledge, SpaceSaving$^\pm$ is the first deterministic algorithm to solve these two problems in the bounded deletion model and make no assumption on the universe. The experimental evaluations of SpaceSaving$^\pm$
highlight that it has the best frequency estimation accuracy among other state of the art sketches, and requires the least space to provide strong guarantees. We also demonstrate that implementing SpaceSaving$^\pm$ with the min and max heap approach provides fast update time. Furthermore, the experiments showcase that SpaceSaving$^\pm$ has very high recall and precision rates across a range of data distributions. These characteristics of SpaceSaving$^\pm$ make it a practical choice
for real world applications. Finally, by leveraging SpaceSaving$^\pm$ and dyadic intervals over bounded universe, we proposed the first deterministic quantile sketch in the bounded deletion model. Our analysis clearly demonstrates that overall, for an unbounded universe or for practical delete:insert ratios below $\frac{logU-1}{logU}$ (e.g., for a realistic universe size of U=$2^{16}$, a ratio of .93 and for U=$2^{32}$, a ratio of .96), SpaceSaving$^\pm$ is the best algorithm to use and solves several major problems with strong guarantees in a unified algorithm.

%\hspace{6em}\textbf{Acknowledgements}

%\noindent We thank the anonymous reviewers for their valuable feedback. 
%This work is funded in part by NSF grants CNS-1703560 and CNS-1815733.

%% file: Appendix.tex
\begin{appendix}

\section{Missing Proofs}

%\subsection{Missing Proofs}

Lemma~\ref{lemma: SS estimation error upper}. \textcolor{black}{\textit{The sum of all estimation error is an upper bound on the sum of frequencies of unmonitored items.}}
\begin{proof}
\textcolor{black}{The SpaceSaving algorithm of $\frac{1}{\epsilon}$ entries with the input stream $\sigma$ can be seen as a collection of entries where each entry process a sub-stream $\iota$, i.e, either increment the item's count, or replace the item with another item and then update the count and the estimation error. The union of all sub-stream becomes the input stream $\sigma$. %The sum of frequencies of all unmonitored item is less or equal to c. 
We want to show that for each entry, its estimation error is an upper bound on the sum of frequencies of items not monitored but assigned to this entry after processing its corresponding $\iota_{entry}$. Hence, the sum of all estimation errors becomes the upper bound on the sum of frequencies of all unmonitored items. }

\textcolor{black}{Consider an arbitrary entry $(x, count_x, error_x)$ and its corresponding sub-stream $\iota$. We want to show that $count_x \geq error_x \geq \forall i \in \iota, \pi(i \ne x)f(i)$. We can proof by induction:
}

\textit{Base case:} \textcolor{black}{Before any input, count is 0, estimation error is 0 and there are no unmonitored items.}

\textit{Induction hypothesis:} \textcolor{black}{After $i$ operations, the entry $(x, count_x, error_x)$ satisfy the relationship: $count_x \geq error_x \geq \forall i \in \iota, \pi(i \ne x)f(i)$.} 

\textit{Induction Step:} \textcolor{black}{Consider the case when the $(i+1)^{th}$ input arrives. If the newly inserted item is $x$, then the count will increase by $1$ and hence the inequality chain still holds. If the newly inserted item is $y$ which is different from the current item, then $x$ is replaced by $y$. The count becomes $count_x + 1$, and estimation error become $count_x$. Since $count_x + 1$ is also the total number of item seen by this entry, then in worst case $y$ only appeared once and hence the sum of frequencies of all unmonitored item from $\iota$ is at most $count_x$ which is the new estimation error. Therefore, $count_y = count_x + 1 \geq count_x = error_y\geq \forall i \in \iota, \pi(i \ne y)f(i)$.}

\textit{Conclusion:} \textcolor{black}{By the principle of induction, for each entry its estimation error is an upper bound of the sum of frequencies of items not monitored but assigned to the entry.}

\textcolor{black}{Based on the induction proof, we know the sum of all estimation errors is the upper bound on the sum of all the frequencies of the items not monitored by but assigned to each individual entry. Moreover, the sum of frequencies of all unmonitored item upper bounded by the sum of all the frequencies of the items not monitored by but assigned to each individual entry after processing its corresponding $\iota$. Hence, the sum of all estimation errors is the upper bound on the sum of frequencies of unmonitored items.}
\end{proof}

Lemma~\ref{lemma error bound}. \textcolor{black}{\textit{SpaceSaving with $O(\frac{1}{\epsilon})$ space can estimate any items' frequency with an additive error less than $\epsilon I$.}}
\begin{proof} Proof By Induction.

\textit{Base case:} \textcolor{black}{Before any insertions, all items have frequency of 0 and all items have estimated frequency of 0.}

\textit{Induction hypothesis:} \textcolor{black}{After $i < I$ insertions, the maximum frequency estimation error of the sketch less than $\epsilon I$.}

\textit{Induction Step:} \textcolor{black}{Consider the case when the $(i+1)^{th}$ insertion arrives. If the newly inserted item $x$ is monitored or the sketch is not full, then no error is introduced. If the newly inserted item $x$ is not monitored and the sketch is full, then $x$ replaces the $minItem$ which has the the minimum associated count, $minCount$. The $minCount$ is upper bounded when every items inside the sketch has the same count, and hence $minCount \leq i\frac{\epsilon}{1}<\epsilon I$. The estimated frequency for $x$ is $minCount$+1 and x is at most overestimated by $minCount$. The frequency estimation for $minItem$ becomes 0, and $minItem$'s frequency estimation is underestimated by at most $minCount$. Therefore, the $error_{max}$ after processing the newly inserted item is still less than $\epsilon I$. }

\textit{Conclusion:} \textcolor{black}{By the principle of induction, SpaceSaving using $O(\frac{1}{\epsilon})$ space solves the frequency estimation problem with bounded error, i.e, $\forall i, |f(i) - f'(i)| < \epsilon I$.}

\end{proof}

\end{appendix}